\documentclass[a4paper, UKenglish,cleveref, autoref, thm-restate,authorcolumns]{lipics-v2021-Supra}
\usepackage{verbatim}
\usepackage{lipsum}
\usepackage{multicol}
\usepackage{amsmath}
\usepackage{framed}
\usepackage{float}
\usepackage{tikz}
\usepackage{soul}
\usepackage{tabularx}

\usepackage{xcolor}
\definecolor{pdfbgcolor}{RGB}{180,180,180}
\definecolor{trevcolor}{RGB}{0,80,200}
\definecolor{dkgreen}{rgb}{0,0.6,0}
\definecolor{mauve}{rgb}{0.58,0,0.82}
\definecolor{nblue}{RGB}{46, 134, 193 }
\definecolor{ngr}{RGB}{142, 68, 173}
\definecolor{col1}{HTML}{7D3C98}
\definecolor{col2}{HTML}{bd910f}
\definecolor{algo1color}{RGB}{0, 128, 0}
\definecolor{algo2color}{RGB}{255, 0, 0}

\usepackage{listings}
\usepackage{lipsum}
\lstdefinestyle{customc}{
    belowcaptionskip = 1\baselineskip,
    breaklines       = true,
    frame            = L,
    xleftmargin      = \parindent,
    language         = C,
    showstringspaces = false,
    escapeinside     = {//}{\^^M},
    basicstyle       = \LSTfont, 
    keywordstyle     = \bfseries\color{green!40!black},
    commentstyle     = \itshape\color{gray!60!black},
    identifierstyle  = \color{blue!50!black},
    stringstyle      = \color{orange},
    otherkeywords    = {then,word,process_local,type,xbegin,xabort,xend},
    numbersep        = 5pt, 
    numberstyle      = \tiny\color{black}, 
    numbers          = left, 
}
\usepackage[linesnumbered, ruled, vlined]{algorithm2e}
\SetAlgoNlRelativeSize{0} 
\SetInd{0em}{1em}         

\SetCommentSty{mycommfont}
\usepackage[skip=4pt]{caption}
\newif\ifcode
\codefalse
\codetrue
\usepackage{macros}
\ifcode

\newcommand{\pscomment}[1]{\textcolor{blue!50}{PS: #1}}

\newcommand{\rwstm}{RWSTM\xspace}
\newcommand{\ostm}{OSTM\xspace}

\newcommand{\preset}{preset}
\newcommand{\Preset}{Preset}

\newcommand{\presetSerialization}{preset serialization}

\newcommand{\presetSerializable}{preset serializable}

\newcommand{\rwOblivious}{read-write oblivious}

\newcommand{\RWOblivious}{Read-write Oblivious}
\newcommand{\rwAware}{read-write aware}

\newcommand{\RWAware}{Read-write Aware}

\newcommand{\btmemory}{block transactional memory}

\newcommand{\BTMemory}{Block Transactional Memory}

\newcommand{\BTM}{BTM}

\newcommand{\STM}{STM}

\newcommand{\stmemory}{software transactional memory}

\newcommand{\Blockstm}{Block-STM}
\newcommand{\BlockSTM}{Block-STM}

\newcommand{\BSTM}{Block-STM}
\newcommand{\PEVM}{PEVM}
\newcommand{\SHySTM}{oBTM}     
\newcommand{\saBlockSTM}{oBTM} 
\newcommand{\oBTM}{oBTM}       
\newcommand{\osaBlockSTM}{iBTM}
\newcommand{\iBTM}{iBTM}       
\newcommand{\SSTM}{BTM}        
\newcommand{\sdag}{dBTM}       
\newcommand{\dBTM}{dBTM}       
\newcommand{\saSupraSTM}{dBTM} 
\newcommand{\dApp}{dApp}

\newcommand{\Adaptive}{Adaptive}
\newcommand{\ERC}{ERC-20}

\newcommand{\Cset}{cSet}
\newcommand{\cSet}{cSet}
\newcommand{\account}{state}

\newcommand{\vmState}{VM state}

\newcommand{\state}{state}

\newcommand{\RNum}[1]{\uppercase\expandafter{\romannumeral #1\relax}}

\newcommand{\true}{\mathit{true}}
\newcommand{\false}{\mathit{false}}

\newcommand{\remove}[1]{}
\newcommand{\Wset}{\textit{Wset}}
\newcommand{\Rset}{\textit{Rset}}
\newcommand{\Dset}{\textit{Dset}}

\newcommand{\Read}{\textit{read}}
\newcommand{\Write}{\textit{write}}
\newcommand{\TryC}{\textit{tryC}}
\newcommand{\tryC}{\textit{tryC}}

\newcommand{\ignore}[1]{}

\usepackage[normalem]{ulem}

\usepackage{pgfplots}
\pgfplotsset{compat=1.15} 
\pgfplotsset{
    grid style = {
        dash pattern	= on 4mm off 1mm,
        line cap		= round, black!15,
        line width		= .2pt
    },
}
\pgfplotsset{
    base/.style={
        scale only axis,
        width          = 6in,
        height         = 3in,
        tick align     = inside,
        axis x line*   = bottom,
        axis y line*   = left,
        x label style  = {at = {(0.5,-.12)}},
        y tick label style	= {/pgf/number format/assume math mode},
        enlarge x limits	= 0.11,
    }
}
\pgfplotsset{
    nnc/.style={
        ybar		= 1pt,
        bar width	= 12.5pt, 
        nodes near coords, 
    }
}

\usepackage{amssymb}

\pdfoutput=1 
\hideLIPIcs  

\nolinenumbers

\usepackage{ifthen}
\newboolean{ArXiv} 
\setboolean{ArXiv}{true} 

\usepackage{tocloft}

\EventLogo{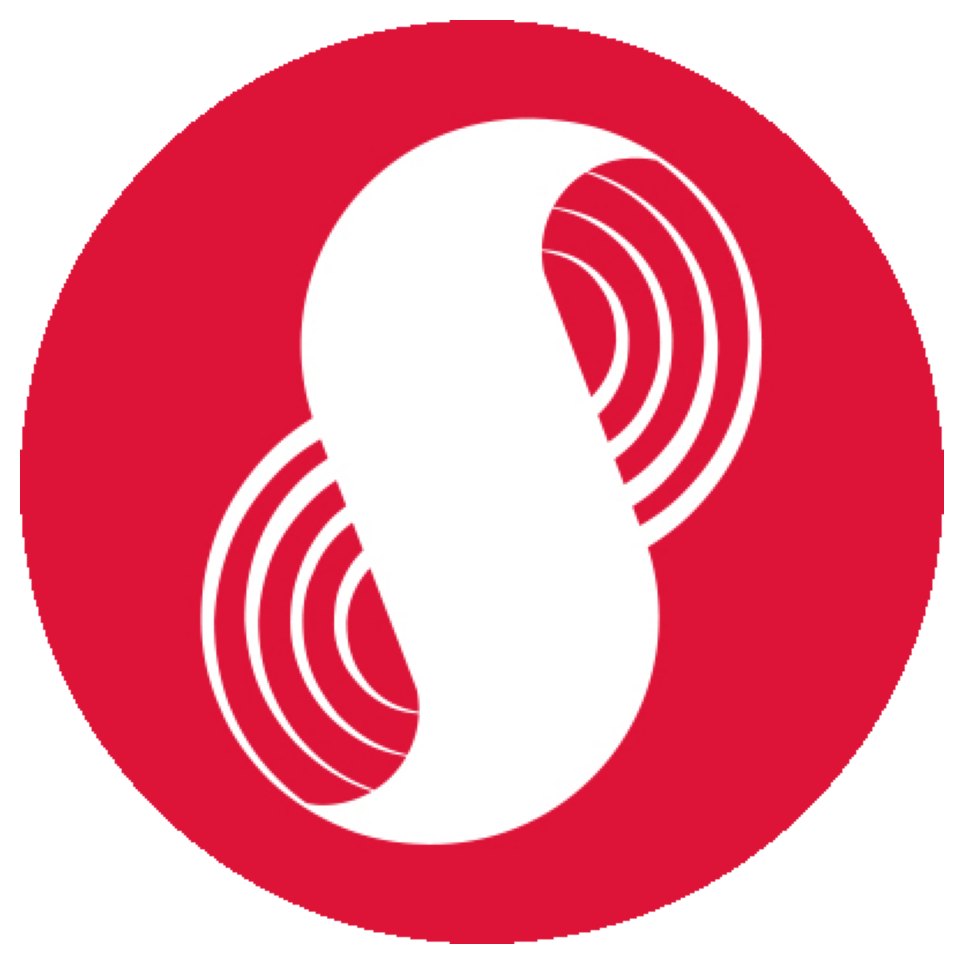}

\tikzset{every picture/.style={scale=1}, font = \Large}

\title{Efficient Parallel Execution of Blockchain Transactions Leveraging Conflict Specifications}
\titlerunning{Efficient Parallel Execution of Blockchain Transactions Leveraging Conflict Specifications}
\author{Parwat Singh Anjana\thanks{Corresponding authors.}}{Supra Research}{p.anjana@supra.com}{0000-0002-6574-3871}{}
\author{Matin Amini}{University of Southern California}{matinami@usc.edu}{0009-0006-8776-3925}{}
\author{Rohit Kapoor}{Supra Research}{r.kapoor@supra.com}{0009-0001-5855-3367}{}
\author{Rahul Parmar}{Supra Research}{r.parmar@supra.com}{0009-0004-4781-9424}{}
\author{Raghavendra Ramesh}{Supra Research}{r.ramesh@supra.com}{0000-0002-6289-9723}{}
\author{Srivatsan Ravi$^1$}{Supra Research \and University of Southern California}{s.ravi@supra.com}{0000-0002-2965-3940}{}
\author{Joshua Tobkin}{Supra Research}{j.tobkin@supra.com}{}{}

\authorrunning{Supra Research}

\newcommand{\revision}[1]{\textcolor{black}{#1}}

\begin{document}
    \maketitle
    \begin{abstract}
Parallel execution of smart contract transactions in large multicore architectures is critical for higher efficiency and improved throughput. 
The main bottleneck for maximizing the throughput of a node through parallel execution is transaction conflict resolution: when two transactions interact with the same data, like an account balance, their order matters. Imagine one transaction sends tokens from account A to account B, and another tries to send tokens from account B to account C. If the second transaction happens before the first one, the token balance in account B might be wrong, causing the entire system to break. Conflicts like these must be managed carefully, or you end up with an inconsistent, unusable blockchain state.

Traditional \stmemory{} (\STM{}) has been identified as a possible abstraction for the concurrent execution of transactions within a block, with \BlockSTM{} pioneering its application for efficient blockchain transaction processing on multicore validator nodes. 
This paper presents a parallel execution methodology that leverages conflict specification information of the transactions
for \btmemory{} (\BTM{}) algorithms. 
Our experimental analysis, conducted over synthetic transactional workloads and real-world blocks, demonstrates that \BTM{s} leveraging conflict specifications outperform their plain counterparts on both EVM and MoveVM. 
Our proposed \BTM{} implementations achieve up to $1.75\times$ speedup over sequential execution and outperform the state-of-the-art Parallel-EVM execution by up to $1.33\times$ across synthetic workloads. 
\end{abstract}

    \tableofcontents
    \markboth{Supra Research}{Efficient Parallel Execution of Blockchain Transactions Leveraging Conflict Specifications}
    
\section{Introduction}\label{sec:intro}   
%
Blockchains are in a race to reduce transaction latency
, improving user experience and increasing throughput to meet anticipated demand. This optimization effort has examined every step of a blockchain transaction, including data dissemination, ordering, execution, and storage. In this paper, we investigate the state-of-the-art in optimizing the execution step of a transaction workflow. 


Blockchains require all nodes to deterministically reach the same final state when executing a block of transactions. To leverage multicore architectures, modern blockchain designs seek to maximize parallel execution throughput while preserving deterministic consistency across all nodes. Traditional \stmemory{} (\STM{}) has been identified as a possible abstraction for the concurrent execution of transactions within a block. However, 
the execution order must follow a \emph{\preset{} order} of transactions in a block.
The main bottleneck for maximizing the throughput of a node through \btmemory{} (\BTM{}) that respects the \preset{} order  is transaction \emph{conflict} resolution: when two transactions, $T_1$ and $T_2$ in a block, interact with the same state, like an account balance, order in which the reads and writes on the state occurs matters for the correctness of the execution. Imagine that a transaction $T_1$ sends tokens from account $A$ to account $B$, and another transaction $T_2$ tries to send tokens from account $B$ to account $C$. If $T_2$ happens before the first one, the token balance in account B might be wrong, causing the entire system to break. Conflicts like these must be managed carefully, or you end up with an inconsistent, unusable blockchain state.

Ethereum maintain consistency by sequentially executing a block of transactions, underutilizing the multicore architectures of its nodes. In contrast, Solana~\cite{solana-url} and Aptos~\cite{aptos} utilize parallel execution, though in different settings. Aptos employs the classical \STM{} technique to execute transactions in parallel with \BTM{}, without requiring any priori read-write access specifications with transactions. 
Aptos’ 
\BlockSTM{}~\cite{blockstm}, represents the current state-of-the-art, applying classical \STM{}~\cite{STM95} techniques to the execution of ordered blocks of transactions. \STM{} techniques are typically speculative or optimistic, meaning that they attempt to execute as many transactions in parallel as possible, while detecting and resolving conflicts. If a transaction reads a value that later changes due to another transaction, it must be re-executed to maintain consistency.
In contrast, Solana's parallel execution requires transactions to be tagged with the storage locations (accounts) read and written, though this is not a \STM{}-based algorithm. We wondered whether the read-write access specifications could enhance a \BTM{} algorithm's performance, and our findings indicated a positive correlation.

Firstly, we observe that some popular blockchains equip transactions a priori with \emph{read-write sets}, i.e., the sets of accounts that the transaction \emph{may} access.
For instance, Solana leverages user-provided read-write sets as access specifications in its lock-profile-based iterative, parallel execution strategy~\cite{solanaSealevel}. Similarly, Sui~\cite{suiParallelExe} utilizes user-provided access specifications to execute transactions in a causally-ordered manner.
Secondly, we further observe that, when such access specifications are not available a priori, it is possible to derive access specifications for public entry functions statically at the time of deployment of smart contracts using data-flow analysis on the smart contract code. This is a one-time computation, enabling efficient utilization of these specifications during transaction execution.
These access specifications essentially serve as \emph{conflict specifications} for any parallel execution technique on ordered blocks of transactions. As a complementary approach, we observe that it is possible to build conflict specifications efficiently at runtime by checking whether the accounts affected by a set of transactions (for instance, \emph{payment transfers} in Aptos Move or \emph{ETH} or \emph{\ERC{}} transfers in Ethereum~\cite{ethereum}) are mutually disjoint.

This is where we see the challenge of improving \BTM{} to further optimize execution times by leveraging transaction conflict specifications. 
Specifically, we seek to answer the following research questions:
Do conflict specifications improve transaction-execution throughput in existing \BTM{} techniques? If so, what is the best \BTM{} technique to optimally leverage these conflict specifications?
By constructively answering these questions, we aim to advance the frontiers of parallel transaction execution, driving significant improvements in blockchain scalability and efficiency.

\noindent\textbf{Contributions.}
This paper presents a holistic methodology for efficient parallel execution of blockchain transactions.
\begin{enumerate}
    \item Leveraging our insight that conflict specifications must be exploited for maximizing throughout in parallel execution, we detail two \BTM{} algorithms, \saSupraSTM{} and \saBlockSTM{}, that can efficiently execute block transactions when \emph{sound}, but possibly \emph{incomplete} specifications are available.
    
    \item We present implementations of our algorithms for blockchains like Ethereum VM (EVM) and Aptos MoveVM and conduct a detailed empirical analysis on real-world blocks, thoroughly analyzing the best and worst-case performance. Additionally, we design a workload generator for analyzing performance on large blocks and unconventional transactional workloads. 
    
    \item We present a rigorous formalism and implementation for how conflict specifications are efficiently constructed for EVM and MoveVM. We then present an integrated implementation (\osaBlockSTM{}) that outperforms existing parallel execution approaches.
    
    \item We evaluated the performance of integrated implementation against sequential execution and state-of-the-art Parallel-EVM (\PEVM{})~\cite{pevm} on both synthetic and historical blocks. In synthetic workload, the proposed \iBTM{} achieve up to $1.75\times$ speedup over sequential and $1.33\times$ over \PEVM{} and maintain consistent improvements, with an average speedup of $1.24\times$ over \PEVM{}. It achieves a maximum speedup of $2.43\times$ over sequential execution on historical workloads. 

\end{enumerate}
The methodology is rigorously demonstrated for EVM and MoveVM, for which conflict specifications are algorithmically derived as part of our integrated implementation. However, we remark that the purpose of our modular presentation and detailed ablation studies is to demonstrate how easily our parallel execution algorithms can be applicable to other blockchain ecosystems, specifically Sui and Solana, which explicitly provide a priori transactional access sets that can be used to derive conflict specifications. By identifying dependencies upfront, we can avoid unnecessary rollbacks and retries, making our system much more efficient for prevalent blockchain infrastructures. This allows the system to automatically separate transactions that need special handling due to dependencies, making the entire execution process more streamlined and as we demonstrate, highly efficient. {Finally, we present how this approach also unleashes the possibility of workload \emph{adaptive} execution that can leverage a sequential or parallel execution algorithm, depending on the block's ``conflict threshold''.}

\noindent\textbf{Roadmap.}
The rest of the paper is organized as follows. \Cref{sec:related,sec:motivation} presents related work and motivation on parallel execution of block transactions, respectively. \Cref{sec:model-algo} introduces the system model and the proposed \BTM{} algorithms (\saSupraSTM{} and \saBlockSTM{}). \Cref{Sec:experiments} presents experimental results on parallel execution using conflict specifications. \Cref{sec:spec} formalize and implement conflict specification generation and evaluate the full integration in both real-world and synthetic workloads. \ifboolexpr{bool{ArXiv}}{}{For succinctness of presentation, we focus on the EVM in the paper and only summarize the MoveVM implementation and evaluation results.}
We conclude in \Cref{sec:conc}.

    \vspace{-5pt}
\section{Related Work}
\label{sec:related}
%
This section reviews smart contract execution models and existing parallel execution algorithms for block transactional execution.
%

Several models have been proposed for the execution of transactions in blockchains; one such model involves the \emph{block proposer} executing transactions, generating a block containing state differences
, and subsequently propagating this block across the network for validators to verify. Another approach entails the block proposer determining the order of transactions in a block, after which all nodes reach consensus on this order before executing transactions in parallel. The former is known as the \emph{Ethereum model}~\cite{ethereum}, while the latter is known as the \emph{Aptos model}~\cite{aptos}. Additionally, a third model facilitates parallel execution by incorporating read-write sets with transactions; this is commonly referred to as the \emph{Solana model}~\cite{solana-url}. A fourth model exploits resource ownership to enable parallel execution, a methodology known as the \emph{Sui model}~\cite{suiDoc}.  

These four execution models can be broadly categorized into two distinct classes that aim to optimize transaction execution
. The first class is called the \emph{\rwAware{} execution} which relies on transaction access hints provided by clients to facilitate parallel execution either through preprocessing in the form of a directed acyclic graph (DAG) or runtime techniques based on lock profiling. The second class comprises techniques that leverage run-time execution techniques, such as \STM{} or locks, to optimistically execute transactions. We call this approach the \emph{\rwOblivious{} execution}. 

\ifboolexpr{bool{ArXiv}}{
\revision{\Cref{fig:related_work} situates existing parallel execution approaches for blockchain transactions within this taxonomy. The \rwAware{} class utilizes access specifications, in the form of read-write sets, to construct lock-based or DAG-based schedules, enabling high parallelism in blockchain networks such as Sui~\cite{suiDoc} and Solana~\cite{umbraresearch}. However, these approaches depend on transactions being annotated with client-provided access specifications, which are often inaccurate due to delays between their generation and the actual time of execution, a limitation that is evident in the higher failure rate of non-voting transactions observed in Solana blocks. Moreover, they introduce additional bandwidth overhead to disseminate this extra metadata with transactions across the network. In contrast, techniques in the \rwOblivious{} class operate without prior access specification knowledge. These systems rely on runtime conflict detection, using optimistic or pessimistic \btmemory{} executions or lock-based techniques to identify and resolve conflicts, examples include \BlockSTM{}~\cite{blockstm}, and other techniques~\cite{Anjana:OptSC:PDP:2019,anjana2022optsmartDAPD22022,Dickerson+:ACSC:PODC:2017,polygonPosParallelization,VikramHerlihy:EmpSdy-Con:Tokenomics:2019,seigiga}. These approaches rely on identifying and resolving conflicts at runtime through speculative execution. Under highly conflicting workloads, these could incur significant overhead that may outweigh the benefits of parallelism, potentially leading to worse performance than sequential execution.}

\begin{figure}[!t]
    \centering
    \scalebox{.60}{\input{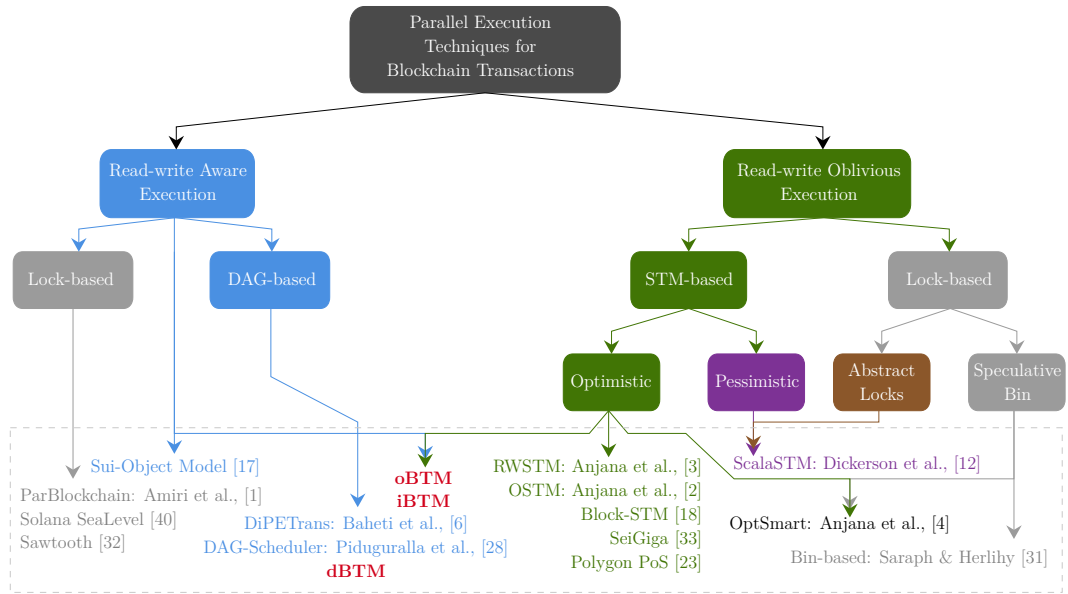}}
    \caption{State-of-the-art techniques for parallel execution of blockchain transactions.}
    \label{fig:related_work}
\end{figure}

}{}

\noindent
\textbf{\RWAware{} Execution.} Transactions on blockchains such as Solana~\cite{solana-url} and Sui~\cite{suiDoc} upfront specify the accounts they access in read and write mode during execution. These access hints are used to enable parallel execution, either through static analysis or by employing a runtime scheduler that resolves read-write set conflicts and executes transactions in parallel. Alternatively, a DAG can be constructed from the read-write access sets to partition transactions into independent groups (iterations in Solana) for parallel execution. 
Specifically, Solana's~\cite{solana-url} SeaLevel~\cite{umbraresearch,solanaSealevel}, leverages read-write sets and lock profiling to execute transactions iteratively. Each iteration involves a locking phase to detect conflicts and an execution phase where non-conflicting transactions are executed in parallel, while conflicting transactions are deferred to subsequent iterations until all transactions in the block are processed. The iteration information is then included in the block by the block proposer to support parallel execution during validation at the validators.



In contrast, Sui~\cite{suiDoc} introduces an object-based state model to identify independent transactions. Objects are shared or exclusively owned, each with a unique identity and owner address. Dependency identification is simplified by determining whether multiple transactions access the same shared object. Based on this model, Sui enables parallel execution~\cite{caseofppchains2022,suiParallelExe} where transactions on owned objects are executed in parallel and can completely bypass consensus, while those on shared objects run sequentially through consensus to avoid conflicts. 
Other works, such as ParBlockchain~\cite{amiri2019parblockchain} and Hyperledger Sawtooth~\cite{sawtooth}, employ lock-based or conflict-analysis techniques, while DiPETrans~\cite{Baheti-DiPETrans-CCPE2022} and the efficient scheduler for Sawtooth~\cite{piduguralla2023dag} utilize DAG-based \rwAware{} execution.


\noindent\textbf{\RWOblivious{} Execution.} 
In this class, the goal is to execute an ordered set of transactions in parallel as if they were executed sequentially and arrive at the same state. The key idea is that some transactions may not be conflicting, i.e, they do not read or write any common \account{s}; hence, can be executed in parallel, enabling execution acceleration that arrives at the correct sequential result.

For the Ethereum model, Dickerson et~al.~\cite{Dickerson+:ACSC:PODC:2017} propose, the first pioneering work on parallel execution of blockchain transactions. They proposed using pessimistic Scala\STM{} for parallel execution at the block proposer and a happen-before graph for parallel execution at the validators. Later, Anjana et~al.~\cite{Anjana-ObjSC-Netys-2020,Anjana:OptSC:PDP:2019} proposed an optimistic \STM{} (\ostm{}) based multi-version timestamp ordering protocol for parallel execution at the proposer, while a DAG-based efficient parallel execution at validators. Saraph and Herlihy~\cite{VikramHerlihy:EmpSdy-Con:Tokenomics:2019} proposed a simple \emph{bin-based two-phase} approach. In the first phase, the proposer uses locks and tries to execute transactions concurrently by rolling back those that lead to conflict(s). Aborted transactions are kept in a sequential bin and executed sequentially in the second phase. 
Later, OptSmart~\cite{anjana2022optsmartDAPD22022} proposed an approach that combines the idea of bin-based approach with the \ostm{} approach for efficient parallel execution.

In the Aptos model, differed execution of transaction based on \presetSerializable{} and multi-version concurrency control is proposed in the \Blockstm{}~\cite{blockstm}. Rather than speculatively executing block transactions in any order, they employ it on ordered-sets, called the \emph{\preset{} order}. Each validator uses \Blockstm{} independently to execute a leader proposal of an ordered set of transactions in parallel to get the same state. This has currently been implemented on the Aptos blockchain~\cite{aptos} and is the most promising approach, as it does not require additional information to be attached to the transactions or in the block for parallel execution. It has been adopted for execution on the Polygon PoS Chain~\cite{polygonPosParallelization}, where it is already live on the mainnet. The proposer uses \Blockstm{} to execute transactions in parallel and includes a DAG in the block to allow deterministic and safe parallel execution at validators. Recently, SeiGiga~\cite{seigiga} has also adopted \Blockstm{}.

\revision{
In contrast to state-of-the-art parallel execution approaches, we propose two \BTM{} approaches, \saSupraSTM{} and \saBlockSTM{}, which take conflict specifications as input and execute transactions efficiently by minimizing aborts and re-execution overhead. We then present a conflict analyzer integrated version of the \saBlockSTM{} approach, which we call the \iBTM{}. The \iBTM{} derives access specifications by leveraging data already available in block transactions and deployed smart contracts, without incurring additional bandwidth overhead or conflict-specification generation costs that cannot be offset by the benefits of parallelism. The proposed \iBTM{} algorithm efficiently executes block transactions leveraging conflict specifications generated by the conflict analyzer. The proposed techniques combine the advantages of both the \rwAware{} and \rwOblivious{} execution models while ensuring \preset{} serialization of block transactions.}
    \section{Motivation and Overview}\label{sec:motivation}
This section motivates the need for parallel execution in blockchain systems and outlines the benefit of leveraging conflict specification for parallel execution. 

\begin{figure}[!t]
    \centering
    \scalebox{.72}{\input{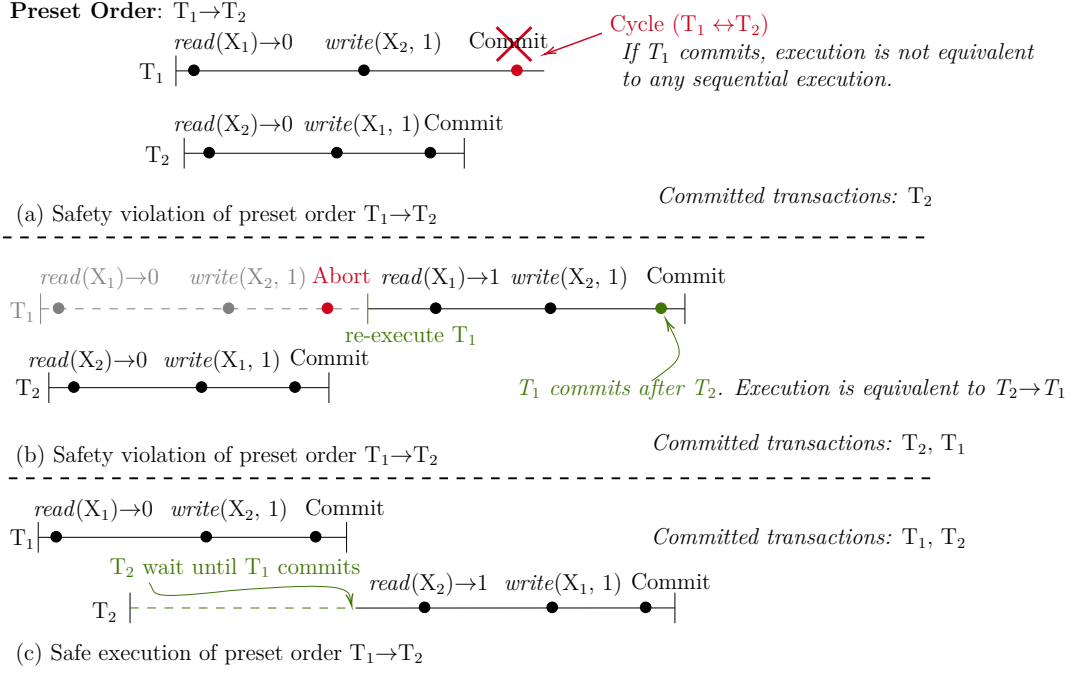}}
    \caption{Safe execution of transactions in \preset{} order: sub-figure (a) illustrates that when conflicting transactions T$_1$ and T$_2$ execute and commit in an arbitrary order, it will result in an unsafe execution; (b) illustrates that transactions commit in some serialization order other than \presetSerialization{}, resulting in an unsafe execution. In sub-figure (c), to ensure safe execution, T$_2$ waits for T$_1$ to commit before committing.}
    \label{fig:safe-exe}
\end{figure}

\noindent\textbf{Motivation.}
Naturally, the throughput of \STM{}-based execution varies widely based on the level of transaction conflicts, aborts, and re-executions (triggered re-validations of the transaction's read state), presenting both a challenge and an opportunity to further optimize execution time.

Consider the scenario shown in \Cref{fig:safe-exe}, in which we have two transactions $T_1$ and $T_2$ running concurrently. We consider that $T_1$ must precede $T_2$ in the \preset{} order (denoted $T_1\to T_2$). Without \emph{a priori} knowledge of read-write conflicts in the set of accounts (\account{s}) accessed, committing $T_2$ prior to committing $T_1$ may result in a safety violation. To illustrate this, consider the following execution: $T_1$ reads an account $X_1$ (value $0$), following which $T_2$ reads $X_2$ (value $0$), followed by a write of a new value $1$ to $X_1$. Observe that if $T_2$ commits at this point in the execution and if $T_1$ writes a new value $1$ to $X_2$ after the commit of $T_2$, then the resulting execution does not respect the \preset{} order in any extension. Clearly, if $T_1$ commits, the resulting execution is not equivalent to any sequential execution, as shown in \Cref{fig:safe-exe}a. Alternatively, if $T_1$ aborts and then re-starts, any read of $X_2$ will return the value $1$ that is written by $T_2$, thus not respecting the \preset{} order, as illustrated in \Cref{fig:safe-exe}b. Consequently, the only possible way to avoid this, requires $T_2$ to wait until $T_1$ commits (see \Cref{fig:safe-exe}c). Now consider a modification of this execution in which $T_2$ reads $X_3$ and writes to $X_4$. In this case, $T_2$ does not have a \emph{read-from conflict} with $T_1$ allowing $T_1$ and $T_2$ to run in parallel with almost no synchronization. 
However, this is possible only if threads that execute these transactions are \emph{aware} of the conflict prior to the execution.

After studying both read-write oblivious and aware models (taxonomized in \cref{sec:related}), we conclude that the approach that leverage conflict specifications is the best way forward. It can be considered the ``Goldilocks'' approach which creates conflict specifications to overcome the limitations of parallel execution imposed by the \preset{} order of transaction. This is what will allow us to scale with parallel execution efficiently and maximize throughput.


\noindent\textbf{Overview.} 
Consider the execution of four transactions T$_1$$-$T$_4$ depicted in \Cref{fig:sece}, as a running example. 
In \Cref{fig:sece}a, we observe a sequential execution, despite T$_1$ and T$_2$ accessing different \state{s}, limiting throughput. This is, for instance, the case in Ethereum~\cite{ethereum} where transactions are executed sequentially. In contrast, \Cref{fig:sece}b illustrates parallel execution, where T$_1$, T$_2$ and T$_3$ execute in parallel and improve throughput. However, T$_4$ has a \emph{reads-from} conflict with T$_2$ on X$_3$ and therefore must wait for T$_2$ to commit. 
Moreover, T$_4$ has a \emph{write} conflict with T$_3$ and must either wait for T$_3$ to commit or employ data structures that allow tracking of both writes performed to X$_5$ (\`a la \emph{multi-versioning}~\cite{tm-book}).

\begin{figure}[!t]
    \centering
    \scalebox{.6}{\input{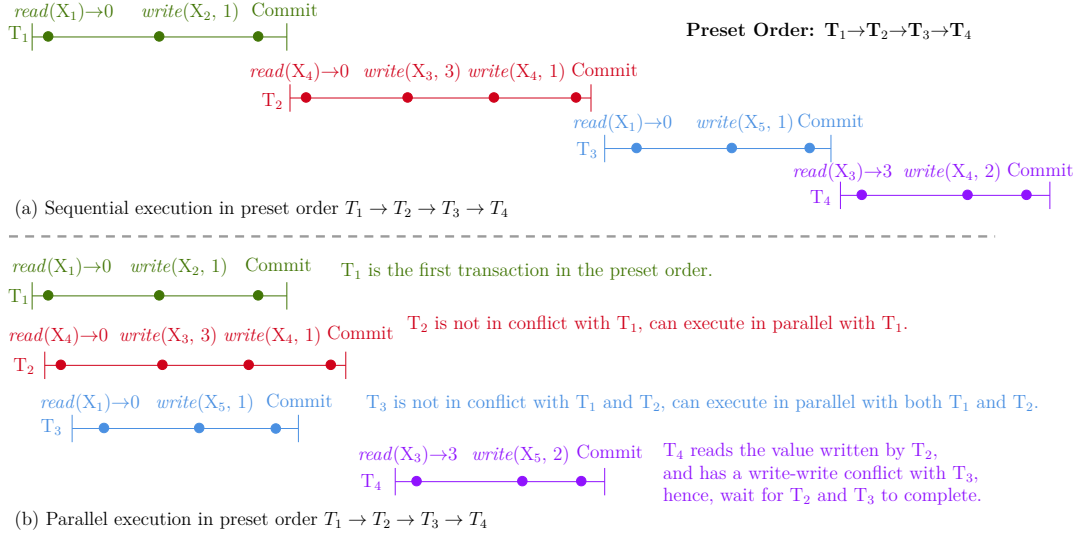}}
	\caption{Leveraging conflict specifications a priori for parallel execution.}
	\label{fig:sece}
\end{figure}

As we constructively observe, the \state{s} accessed by these transactions can be efficiently derived a priori to build \emph{sound}, but possibly incomplete specifications. Depending on how complete the derived specifications are, this information can be used to minimize runtime transactional conflicts and thus avoid unnecessary transaction aborts during parallel execution. However, as the results in this paper show, not only are there rigorous theoretical approaches to deriving conflict specifications in smart contract ecosystems like Ethereum, fast techniques exist to identify a non-trivial number of transactional conflicts within a block, which provide significant speedup for parallel execution.

    \section{\BTM{} Execution From Conflict Specifications}\label{sec:model-algo}
In this section, we formally present \btmemory{} in the asynchronous shared memory model, detail our \BTM{} algorithms that leverage conflict specification and satisfy \preset{} serializability.
\subsection{Model}
The model of \BTM{} in this paper is presented in the standard asynchronous shared memory model~\cite{Lyn96}.

\noindent\textbf{Transactions.} A \emph{transaction} is a sequence of \emph{transactional operations}, reads and writes, performed on a set of virtual machine (VM) states. A \BTM{}[1,\ldots, n] \emph{implementation} provides a set of
concurrent \emph{processes} with deterministic algorithms that implement reads and
writes on accounts using a set of \emph{shared memory locations} accessed by the $n$ transactions with a \preset{} order $T_1 \to , \ldots, T_n$.
More precisely, for each transaction $T_k$, a \BTM{} implementation must support the following operations: 
\Read$_k$($X$), where $X$ is an object, that returns a value in
a domain $V$
or a special value $A_k\notin V$ (\emph{abort}),
\Write$_k$($X$, $v$), for a value $v \in V$,
that returns $\mathit{ok}$ or $A_k$, and
$\mathit{\TryC}_k$ that returns $C_k\notin V$ (\emph{commit}) or $A_k$.
The transaction $T_k$ completes when any of its operations returns $A_k$ or $C_k$.

\noindent\textbf{Executions, Histories and Conflicts.}
A \BTM{} \emph{execution} is a sequence of \emph{events} performed on shared memory \state{s} by an interleaving of transactions as prescribed by the implementation. To avoid introducing additional technical machinery that is not strictly necessary to follow the algorithmic exposition in this paper, we do not define executions using configuration semantics as is common in traditional shared memory systems~\cite{tm-book,Lyn96}, although it would be straightforward to do so. 

A \BTM{} \emph{history} is the subsequence of an execution consisting of the invocation and response events of operations of the transactions. For a transaction $T_k$, we denote all objects accessed by its read and write as \Rset($T_k$) and \Wset($T_k$), respectively. We denote all the operations of a transaction $T_k$ as \Dset($T_k$). The \emph{read set} (resp., the \emph{write set}) of a transaction $T_k$ in an execution $E$, denoted \Rset$_E$($T_k$) (resp. \Wset$_E$($T_k$)), is the set of objects that $T_k$ attempts to read (resp. write) by issuing a read (resp. write) invocation in $E$. The \emph{data set} of $T_k$ is \Dset($T_k$) = \Rset($T_k$)$\cup$\Wset($T_k$). $T_k$ is called \emph{read-only} if \Wset($T_k$) $=\emptyset$; \emph{write-only} if \Rset($T_k$)$\ =\emptyset$ and \emph{updating} if \Wset($T_k$) $\neq\emptyset$.

We say 
$T_i$ and $T_j$ \emph{conflict} in an execution $E$ if there exists a common \state{} $X$ in \Dset($T_i$) and \Dset($T_j$) such that $X$ is contained within \Wset($T_i$) or \Wset($T_j$), or both. Furthermore, we say that $T_i, T_j$ \emph{read-from conflict} if $\Wset(T_i) \cap \Rset(T_j) \neq \emptyset$ and $T_i$ appears before $T_j$ in the \preset{} order. Note that the definition of read-from conflict, unlike that of a conflict, relies on a \preset{} order existing between the two transactions.

Let $H$ be a 
sequential history, i.e., no two transactions are concurrent in $H$.
For every operation \Read$_k$($X$) in $H$,
we define the \emph{latest written value} of $X$ as follows:
if $T_k$ contains a \Write$_k$($X$, $v$) that precedes \Read$_k$($X$),
then the latest written value of $X$ is the value of the latest such write to $X$.
Otherwise,
the latest written value of $X$ is the value
of the argument of the latest \Write$_m$($X$, $v$) that precedes
\Read$_k$($X$) and belongs to a committed transaction in $H$.
This write is well-defined since $H$ can be assumed to start with an initial transaction writing to all \account{s}. 
We say that a sequential history $S$ is \emph{legal} if every read of a \account{} returns the \emph{latest written value} of this \account{} in $S$. It means that sequential history $S$ is legal if all its reads are legal.

\vspace{2pt}
\noindent\textbf{\Preset{} Serializability.}
Given a set of $n$ transactions with a \emph{\preset{}} order $T_1 \to T_2 \to... \to T_n$, we need a deterministic parallel execution protocol that efficiently executes block transactions using the serialization order and always leads to the same state, even when executed sequentially. We formalize this using the definition of \emph{\preset{} serializability}. 
\begin{definition}[\Preset{} serializability]\label{def:presetserializability}
Let $H$ be a history of a \BSTM{}[1,\ldots, n] implementation $M$. We say that $H$ is \emph{\presetSerializable{}} if $H$ is equivalent to a legal sequential history $S$ that is $H_1 \cdots H_{i}\cdot H_{i\text{+}1} \cdots H_n$ where $H_i$ is the complete history of transaction $T_i$.  
We say that a \BSTM{}[1,\ldots, n] implementation $M$ is \emph{\presetSerializable{}} if every history of $M$ is \presetSerializable{}. 
\end{definition}




\ifthenelse{\boolean{ArXiv}}{

}

\subsection{Algorithm overview}\label{sec:algo_overview}
The core idea behind our algorithm is that only independent transactions are executed in parallel, ensuring that no race conditions arise during execution. For any transaction $T_k$, if there exists a transaction $T_i$ $\not\in$ \Cset{(T$_k$)} that precedes $T_k$ in the \preset{} serialization order, the execution of $T_k$ is deferred until $T_i$ completes (here \Cset{(T$_k$)} denotes the set of transactions that do not conflict with $T_k$). Additionally, after executing each transaction, validation is performed to ensure that if any specification for a transaction is incorrect and two dependent transactions are executed in parallel, the transaction higher in the \preset{} serialization order can abort and re-execute. Thus, the output of the algorithm will be the same as that of a sequential execution.

Performance improvement over parallel execution techniques based on \BSTM{} is achieved by reducing the number of aborts and re-executions. This is made possible by leveraging the knowledge of the transaction conflict specifications: a transaction $T_k$ is executed only after ensuring that all preceding transactions in the \preset{} serialization order belong to its independence set, \Cset{(T$_k$)}. This targeted execution strategy minimizes conflicts and improves throughput.

\noindent\textbf{Problem Statement.}
In this section, we ask the following question: given a set of $n$ transactions $T_1\to,\ldots,\to T_n$ and \Cset{(T$_k$)} for all $k\in \{1,\ldots, n \}$, what are the most efficient algorithms for implementing 
\BSTM{}[1,\ldots, n]. 
We implemented two different algorithms, leveraging conflict specification in 
DAG \SSTM{} (\saSupraSTM{}) and 
optimized \BTM{} (\saBlockSTM{}), based on how the \emph{scheduler} utilizes the transaction conflict specification for efficient parallel execution. As explained in \cref{sec:intro,sec:related}, this mechanism for implementing \BTM{} is applicable directly to the \rwAware{} models like Sui and Solana in which the \BTM{} is explicitly provided the transactional read-write sets. More importantly, as we demonstrate in \cref{sec:spec}, we can efficiently construct conflict specifications for \rwOblivious{} models like Ethereum EVM and Aptos MoveVM, even if under-approximate (i.e., incomplete), and still reap the benefits of our proposed methodologies for implementing \BTM{}. 

\ifthenelse{\boolean{ArXiv}}
  {
\paragraph*{Leveraging Conflict Specification in DAG \SSTM{} (\saSupraSTM{})}  

This approach utilizes the conflict specifications of the transactions (which are assumed to be correct or overapproximated) and the \preset{} order to create a DAG, which is used as a partial order and serves as input for the scheduler-less execution algorithm. In the DAG, transactions are represented as vertices, whereas conflicts among transactions are denoted as directed edges. The \emph{indegree} field is added with each vertex (transaction) to track dependencies with prior transactions in the \preset{} order; a transaction becomes eligible for execution when its indegree is zero. 
A transaction is considered independent if its \Cset{} includes all transactions that precede it in the \preset{} order and will not have edges from any preceding transactions, resulting in an indegree of zero. During execution, the non-zero indegree transactions wait for the preceding transactions to commit and clear dependencies. 
At the time of commit, the committing transaction decreases the indegree of all dependent transactions. 

We now present the implementation of the \preset{}-serializable \saSupraSTM{} (\Cref{algo:sdag}), which uses transaction conflict specifications to construct a DAG representing a \preset{}-serializable partial order.


\noindent\textbf{Implementation state.} The \emph{DAG} is implemented using two primary data structures: 1. \textit{indegree}, where 
\textit{indegree[$k$]} denotes the number of preceding transactions in the \preset{} order that transaction $T_k$ depends on; and 2. \textit{dependents}, which maps each $T_k$ to a set of its dependent successors. 
For each \emph{\vmState{} $X_i$}, the algorithm maintains a memory location $v_i$ that stores a set of tuples $([v_1,k],[v_2,k'],\ldots)$. Each tuple $[v,k]$ represents the value $v$ of $X_i$ and $k$ is the transaction $T_k$ that wrote that value.

\begin{algorithm}[!tb]
\footnotesize
\caption{\saSupraSTM{}[1,\ldots, n]: It is a DAG-based scheduler to execute independent transactions in parallel. The indegree field is added with each transaction to track dependencies with prior transactions. Consider a transaction $T_k$ being executed by a process $p_k$. 
}
\label{algo:sdag}
\KwIn{${T}$: list of transactions in the block $B_i$; ${S}$: pre-state$-$ state before execution of block $B_{i}$; \\$\Cset{}$: specifications for transactions in $T$.
}
\KwData{$indegree[T_k]$: the number of transactions in \preset{} order $T_i$ is dependent on;\\ $dependents[k]$: the set of transactions dependent upon $T_k$; a version list $\langle v_j \rangle$ for each \account{} $X_j$.
}
\noindent
\begin{minipage}[t]{.46\textwidth}
\LinesNotNumbered
\SetNlSty{}{}{}

\SetKwProg{Pn}{Fun}{:}{}
\SetKwFunction{readk}{\texttt{\Read$_k$}}
\Pn{\readk{X$_j$}}{\label{alg:dag_read_s}
    \nlset{2} \If{X$_j$ $\not\in$ \Wset(T$_k$)}{\label{alg:dag_read_check}
        
        \tcp{Read the latest version of $X_j$ created by a $T_i$ that precedes $T_k$.}
        
        \nlset{3} [ov$_j$, i] $:=$ \emph{read\_lvp}(T$_k$, X$_j$) \label{alg:dag_read_lvp}
        
        \nlset{4} \Rset(T$_k$) $:=$ \Rset(T$_k$) $\cup$ \{X$_j$, [ov$_j$, i]\} \label{alg:dag_read_rset}
        
        \nlset{5} \Return ov$_j$ 
    }
     \nlset{6} \Else{
        \tcp{$X_j$ is in \Wset(T$_k$)} 
        
        \nlset{7} [ov$_j$, $\bot$ ] $:=$ \Wset(T$_k$).$\lit{locate}$(X$_j$) \label{alg:dag_read_wset}
    
        \nlset{8} \Return ov$_j$ \label{alg:dag_read_e}
    }
}

\SetKwFunction{lvp}{$\emph{read\_lvp}$}
\SetKwProg{Pn}{Fun}{:}{}
\nlset{9}\Pn{\lvp{T$_k$, X$_j$}}{\label{alg:dag_lvp_s}

    \nlset{10} [ov, i] $:=$ [0, 0];
    
    \tcp{Read the largest version of X$_j$ created by a $T_i$ that precedes $T_k$}
    
    \nlset{11} \ForAll{[ov$_j$, i] $\in$ X$_j$}{
    
        \nlset{12} \If{$k > i$}{
        
            \nlset{13} [ov, i] $:=$ [ov$_j$, i]
        }    
    }
    
    \nlset{14} \Return{[ov$_j$, i]}\label{alg:dag_lvp_e}
}

\SetKwFunction{writemvk}{\Write{$_k$}}
\SetKwProg{Pn}{Fun}{:}{}
\nlset{15}\Pn{\writemvk{X$_j$, v}}{\label{alg:dag_write_s}

    \nlset{16} nv$_j$ $:=$ $v$
    
    \nlset{17} \If{X$_j$ $\not\in$ \Wset(T$_k$)}{\label{alg:dag_write_check}
    
        \nlset{18} \Wset(T$_k$) $:=$ \Wset(T$_k$) $\cup$ \{X$_j$, [nv$_j$, k]\} \label{alg:dag_write_wset}
    }
    \nlset{19} \Else{
    
        \tcp{$X_j$ is in $\Wset(T_k)$, update its current value to $v$.}
        
        \nlset{20} \Wset(T$_k$) $:=$ \Wset(T$_k$).$\lit{update}$(X$_j$, [nv$_j$, k]) \label{alg:dag_write_update}
    }
    \nlset{21} \Return{$ok$} \label{alg:dag_write_e}
}

\end{minipage}%
\hfill
\begin{minipage}[t]{.4\textwidth}
\LinesNotNumbered
\SetNlSty{}{}{}

\SetKwFunction{TryCk}{\tryC{$_k$}}
\SetKwProg{Fn}{Fun}{:}{}
\nlset{22}\Fn{\TryCk()}{\label{alg:dag_tryc_s}
    \tcp{Ensure commit order}
    
    \nlset{23} \If{$indegree[k] \neq 0$}{\label{alg:dag_tryc_check}
    
        \nlset{24} wait until $indegree[k] = 0$;\label{alg:dag_tryc_wait}
    }

    \tcp{Write back to shared memory}
    \nlset{25} \ForAll{X$_j$ $\in$ \Wset(T$_k$)}{\label{alg:dag_tryc_wset}
    
        \nlset{26} Write(X$_j$, [nv$_j$, k])\label{alg:dag_tryc_wset_update}
    }

    \tcp{Clear dependencies}
    \nlset{27}\ForAll{T$_i$ $\in$ dependents(i)}{\label{alg:dag_tryc_indegree}
    
        \nlset{28} $indegree[i] \gets indegree[i] - 1$;\label{alg:dag_tryc_indegree-1}
    }

    \nlset{29} \Return $C_k$ \label{alg:dag_tryc_e}
}

\SetKwProg{Fn}{Fun}{:}{}
\SetKwFunction{genDag}{$gen\_dag$}
\nlset{30}\Fn{\genDag{\Cset{}}}{\label{alg:dag_start}

    \nlset{31}$block\_size \gets size\_of(B_i)$;

    \nlset{32}\For{$k \in $ (0, block\_size - 1))}{\label{alg:dag_init_start}
    
        \nlset{33}$dependencies[k] := \emptyset$
        
        \nlset{34}$dependents[k] := \emptyset$
        
        \nlset{35}$indegree[k] \gets 0$;\label{alg:dag_init_end}
    }

    \tcp{Compute complement sets}
    \nlset{36}\For{$k \in $ (0, block\_size - 1)}{\label{alg:dag_comp_start}
    
        \nlset{37}$\Cset{}Comp(T_k) \gets \{ j \mid j < k \ \text{and} \ j \notin \Cset{(T_k)}\}$;\label{alg:dag_comp_end}
    }

    \tcp{Use \Cset{}Comp to build dependencies}
    
    \nlset{38}\For{$k \in $ (0, block\_size - 1)}{\label{alg:dag_build_cset}
    
        \nlset{39}\ForEach{$j \in \Cset{}Comp(T_k)$}{
        
            \nlset{40}$dependencies[k].add(j)$\label{alg:dag_dep_add}
            
            \nlset{41}$dependents[j].add(k)$\label{alg:dag_depnt_add}
        }
        \nlset{42}$indegree[k] := |dependencies[k]|$\label{alg:dag_indegree_final}
    }

    \nlset{43}\Return $(dependents, indegree)$\label{alg:dag_end}
}
\end{minipage}
\end{algorithm}

\noindent\textbf{Read implementation.} The \textit{\Read{$_k(X_j)$}} function returns the value of \account{} $X_j$ as visible to transaction $T_k$. It first checks whether $X_j$ is already present in \Wset($T_k$) at line~\ref{alg:dag_read_check}. If not, it reads the largest version written by a transaction $T_i$ in shared memory such that $T_i\to T_k$, using the helper function \textit{read\_lvp()} defined in lines~\ref{alg:dag_lvp_s}-\ref{alg:dag_lvp_e}. This value, along with its source, is then recorded in \Rset($T_k$) at line~\ref{alg:dag_read_rset}. 
If $X_j$ has already been locally written by $T_k$, it is returned directly from the \Wset($T_k$) at line~\ref{alg:dag_read_wset}. 

The helper function \emph{read\_lvp}$(T_k, X_j)$ (lines~\ref{alg:dag_lvp_s}--\ref{alg:dag_lvp_e}) is invoked by \Read$_k$ to identify the latest version of $X_j$ in the multi-version data structure created by a transaction that precedes $T_k$ in the DAG. It iterates over all versions of $X_j$ in shared memory and selects the one with the highest index $i < k$. This mechanism ensures multi-version consistency by enforcing serialization semantics across transactions.

\noindent\textbf{Write implementation.} The \textit{\Write}$_k(X_j, v)$ function (lines~\ref{alg:dag_write_s} to \ref{alg:dag_write_e}) stages a write to \vmState{} $X_j$ with value $v$ in the local context of $T_k$. If this is the first write to $X_j$ by $T_k$, an entry is created in \Wset($T_k$). If $X_j$ has been written previously, the existing entry is updated with the new value. This operation does not alter shared memory but prepares the local write set for eventual commit.

\noindent\textbf{Commit implementation.} The \textit{\TryC{$_k$}()} function, lines~\ref{alg:dag_tryc_s} to \ref{alg:dag_tryc_e}, performs the commit logic for transaction $T_k$. It first waits until all dependencies are satisfied, that is, all preceding transactions in the DAG on which $T_k$ depends are committed and $\textit{indegree}[k] := 0$ (line~\ref{alg:dag_tryc_wait}). Then updates the shared memory with writes from $\Wset(T_k)$ at line~\ref{alg:dag_tryc_wset}. At line~\ref{alg:dag_tryc_indegree}, it clears the dependencies of all successor transactions, specifically, for each transaction $T_i > T_k$ where $T_k \notin \textit{\Cset{}}(i)$, it decrements $\textit{indegree}[i]$ by 1 at line~\ref{alg:dag_tryc_indegree-1}. 
Finally, the function returns the commit result $C_k$ at line~\ref{alg:dag_tryc_e}.

\noindent\textbf{DAG Generation.} The \textit{gen\_dag(\Cset{})} function, lines~\ref{alg:dag_start} to~\ref{alg:dag_end}, constructs a DAG from the specification \textit{\Cset{}}, which maps each transaction $T_k$ to the set of preceding transactions (in \preset{} order) that $T_k$ is independent of. It initializes the per-transaction data structures (\textit{dependencies} $:= \emptyset$, \textit{dependents} $:= \emptyset$, and \textit{indegree}$\gets 0$) in lines~\ref{alg:dag_init_start}-\ref{alg:dag_init_end}. Then, computes the complement conflict set \textit{\Cset{}Comp(T$_k$)} for each transaction $T_k$ $\in$ $T$ (line~\ref{alg:dag_comp_start}), which includes all preceding transactions of $T_k$ that are not in \textit{\Cset{(T$_k$)}}. These are the transactions with which $T_k$ conflicts, and hence it must be conservatively dependent on them in the DAG. Using \textit{\Cset{}Comp}, the function constructs the dependency edges (lines~\ref{alg:dag_build_cset}--\ref{alg:dag_indegree_final}). 
For every $T_k$ $\in T$, the function sets \textit{indegree[$k$]} to the cardinality of the dependency set of $T_k$. Finally, the function returns references to the \textit{dependents} and \textit{indegree} structures at line~\ref{alg:dag_end}, which are used as input to DAG-based scheduler for execution.

\vspace{2pt}
\noindent\textit{\underline{Proof of Correctness}.} 
We prove \preset{} serializability for the \Cref{algo:sdag} by adapting the proof structure from traditional \STM{} literature for the block transactional model~\cite{tm-book,KR11}. 

%
\revision{
\begin{lemma}
\label{lm:alg}
\dBTM{}[1,\ldots, n] implementation described in \Cref{algo:sdag} is \preset{} serializable.
\end{lemma}

Let $E$ by any finite execution of the Algorithm. 
Let $<_E$ denote a total-order on events in $E$.
Let $H$ denote a subsequence of $E$ constructed by selecting
\emph{linearization points}
of operations performed in $E$.
The linearization point of a operation $op$, denoted as $\ell_{op}$ is associated with a memory location event or an event performed during 
the execution of $op$ using the following procedure. 
First, we obtain a modification of $E$ by removing 
every incomplete $\Read_k$, $\Write_k$, $\TryC_k$ operation from $E$.

\noindent\textit{Linearization points.}
We now associate linearization points to operations in the obtained completion of $E$ as follows: 
for every read by transaction $T_k$, $\ell_{op_k}$ is chosen as the event associated with the return of \textit{read\_lvp()}.
For every write $op_k$, $\ell_{op_k}$ is chosen as the invocation event of $op_k$.
For every $op_k=\TryC_k$ that returns $C_k$, $\ell_{op_k}$ is associated with line~\ref{alg:dag_tryc_wset}.
$<_H$ denotes a total-order on operations in the sequential history $H$.

\noindent\textit{Serialization points.}
The serialization of a transaction $T_j$, denoted as $\delta_{T_j}$ is
associated with the linearization point of a operation performed during the execution of the transaction.
We obtain a history ${\bar H}$ from $H$ as follows: 
for every transaction $T_k$ in $H$ that is complete, but not completed all its operations, 
we remove it from $H$. 

A complete sequential history $S$ equivalent to ${\bar H}$ is obtained by associating serialization points to transactions in ${\bar H}$ as follows: If $T_k$ is an update transaction that commits, then $\delta_{T_k}$ is $\ell_{\TryC_k}$. If $T_k$ is a read-only transaction in $\bar H$, then $\delta_{T_k}$ is assigned to the linearization point of the last read in $T_k$. Let $<_S$ denote a total-order on transactions in the sequential history $S$.

\begin{claim}
\label{cl:happenbefore}
$S$ is legal. 
\end{claim}
\begin{proof}
First, observe that for every \Read$_j$($X$) $\to v$ in $E$, there exists some transaction $T_i$ that performs \Write$_i$($X,v$) and completes the shared memory write to $v$ during $\TryC_i$.

Consider a \Read$_j$($X$) that returns a response $v$ performed by a transaction $T_j$.
To prove that $S$ is legal, we need to show that,
there does not exist any
transaction $T_k$ that returns $C_k$ in $S$ and performs \Write$_k$\text{(}$X,v'$\text{)}; $v'\neq v$ such that $T_i <_S T_k <_S T_j$. We abuse notation here by assuming $T_i \rightarrow T_k \rightarrow T_j$.
Now, suppose by contradiction that there exists a committed transaction $T_k$, $X \in $\Wset($T_k$) that writes $v'\neq v$ to $X$ 
such that $T_i <_S T_k <_S T_j$.
Since $T_i$ and $T_k$ are both updating transactions (i.e., have non-empty write sets) that commit,
($T_i <_S T_k$) implies that ($\delta_{T_i} <_{E} \delta_{T_k}$) and
($\delta_{T_i} <_{E} \delta_{T_k}$) implies that ($\ell_{\TryC_i} <_{E} \ell_{\TryC_k}$).

Since $T_i$ and $T_k$ are both updating transactions that commit,
($T_i <_S T_k$) implies that ($\delta_{T_i} <_{E} \delta_{T_k}$) and
($\delta_{T_i} <_{E} \delta_{T_k}$) implies that ($\ell_{\TryC_i} <_{E} \ell_{\TryC_k}$).

Now observe that, since $T_j$ reads the value of $X$ written by $T_i$, one of the following is true: 
$\ell_{\TryC_i} <_{E} \ell_{\TryC_k} <_{E} \ell_{\Read_j\text{(}X\text{)}}$ or
$\ell_{\TryC_i} <_{E} \ell_{\Read_j\text{(}X\text{)}} <_{E} \ell_{\TryC_k}$.

Observe that the case that $\ell_{\TryC_i} <_{E} \ell_{\TryC_k} <_{E} \ell_{\Read_j\text{(}X\text{)}}$ is not possible. This is because the value of the \account{} $X$ will have been overwritten by transaction $T_k$ and $\Read_j\text{(}X\text{)}$ should have read the value written by $T_k$ (per \textit{read\_lvp()} ) and not $T_i$---contradiction.
Consequently, the only feasible case is that $\ell_{\TryC_i} <_{E} \ell_{\Read_j\text{(}X\text{)}} <_{E} \ell_{\TryC_k}$.
We now need to prove that $\delta_{T_{j}}$ indeed precedes $\delta_{T_{k}}=\ell_{\TryC_k}$ in $E$.

Now consider two cases:
(1) Suppose that $T_j$ is a read-only transaction. 
Then, $\delta_{T_j}$ is assigned to the last read performed by $T_j$.
If \Read$_j$($X$) is not the last read, then there exists a $\Read_j\text{(}X'\text{)}$ such that 
$\ell_{\Read_j\text{(}X\text{)}} <_{E} \ell_{\TryC_k} <_E \ell_{read_j\text{(}X'\text{)}}$. 
(2) Suppose that $T_j$ is an updating transaction that commits, then $\delta_{T_j}=\ell_{\TryC_j}$ which implies that
$\ell_{read_j\text{(}X\text{)}} <_{E} \ell_{\TryC_k} <_E \ell_{\TryC_j}$. 
Both cases lead to a contradiction. This is because, the commit logic for transaction $T_k$ (lines~\ref{alg:dag_tryc_s} to \ref{alg:dag_tryc_e}). This forces transactions to wait until all dependencies are satisfied, that is, all preceding transactions in the DAG on which $T_k$ depends are committed and $\textit{indegree}[k] := 0$ (line~\ref{alg:dag_tryc_wait}). This would violate the assumption about the return value of \Read$_j$($X$).
\end{proof}
Finally, to complete the proof, we observe that 
if $T_i \to T_j$, then $T_i <_S T_j$.
This follows immediately from the assignment of serialization points. 
}


\paragraph*{Leveraging Conflict Specification in Optimized \BTM{} (\saBlockSTM{})}
The \saBlockSTM{} algorithm optimizes the \BSTM{} scheduler by integrating independence information (i.e., $\Cset{}$) associated with transactions, which is assumed to be complete and provided as input. The scheduler is modified to allow transactions that are independent of all preceding transactions in the \preset{} order to execute without validation. A greater number of such independent transactions improves the execution efficiency. This method builds on the existing scheduler of \BSTM{}~\cite{blockstm} for MoveVM and \PEVM{}~\cite{pevm} developed by the RISE chain~\cite{riselabs} for EVM. 
In the absence of conflict specifications, the algorithm defaults to optimistic execution as in \PEVM{} and \BlockSTM{}.

The design of \saBlockSTM{}, shown in \Cref{algo:saBlockSTM}, leverages the conflict specifications of transactions to enable efficient optimistic execution by minimizing the validation overhead of vanilla \PEVM{} and \BSTM{}. Consider a transaction $T_k$ being executed by a process $p_k$.

\begin{algorithm}[!tb]
\footnotesize
\caption{\saBlockSTM{}[1,\ldots, n]: The approach optimizes the \BlockSTM{} scheduler to allow transactions that are independent of all previous transactions to execute without validation. Consider a transaction $T_k$ being executed by a process $p_k$. 
}
\label{algo:saBlockSTM}
\KwIn{${T}$: list of transactions in the block $B_i$; ${S}$: pre-state$-$ state before execution of block $B_{i}$; $\cSet{}$: specifications for transactions in $T$.
}
\KwData{$\cSet{(T_k)}$: The set of transactions $T_i$ such that $i < k$ and $T_k$ is independent of $T_i$.}

\noindent
\begin{minipage}[t]{.4\textwidth}
\LinesNotNumbered
\SetNlSty{}{}{}

\SetKwProg{Pn}{Fun}{:}{}
\SetKwFunction{readk}{\texttt{\Read{$_k$}}}
\Pn{\readk{X$_j$}}{
    \nlset{2} \If{X$_j$ $\not\in$ \Wset(T$_k$)}{
        \tcp{Read the largest version of $X_j$ created by a tx $T_i$ preceding $T_k$.}
        
        \nlset{3} [ov$_j$, i] $:=$ \emph{read\_lvp}(T$_k$, X$_j$) \label{alg:sabstm_read_lvp}

        \nlset{4} \If{[ov$_{j}$, i].is\_estimate}{
        
           \nlset{5} \If{$\neg$\,\texttt{add\_dependency}(T$_k$, T$_i$)}{\label{alg:sabstm_read_dep}
            
                \nlset{6} \texttt{retry}(T$_k$)\label{alg:sabstm_read_retry}
            }
            \tcp{If T$_k$ read an estimated value written by $T_i$}
            
            \nlset{7} \Return $A_k$
        }
        \nlset{8} \Rset{(T$_k$)} $:=$ \Rset{(T$_k$)} $\cup$ \{X$_j$, [ov$_j$, i]\}\label{alg:sabstm_read_rset}
        
        \nlset{9} \Return ov$_j$
    }
    \nlset{10} \Else{
    
        \nlset{11} [ov$_j$, $\bot$] $:=$ \Wset{(T$_k$)}.$\lit{locate}$(X$_j$) \label{alg:sabstm_read_local}
        
        \nlset{12} \Return ov$_j$
    }
}

\ignore{
\nlset{13} \Pn{\lvp{T$_k$, X$_j$}}{

    \nlset{14} [ov, i] $:=$ [0, 0]
    
    \tcp{Traverse version list of $X_j$ to read the largest version created by a $T_i$ preceding $T_k$.}
    
    \nlset{15} \ForAll{[ov$_j$, i] $\in$ X$_j$}{
    
        \nlset{16} \If{$k > i$}{
        
            \nlset{17} [ov, i] $:=$ [ov$_j$, i]
        }
    }
    \nlset{18} \Return{[ov$_j$, i]}
}
}
\SetKwFunction{writemvk}{\Write$_k$}
\SetKwProg{Pn}{Fun}{:}{}
\nlset{13} \Pn{\writemvk{$X_j, v$}}{\label{alg:sabstm_write_s}

    \nlset{14} ov$_j$ $:= v$

    \nlset{15} \If{X$_j$ $\not\in$ \Wset(T$_k$)}{\label{alg:sabstm_write_wset_check}
    
        \nlset{16} \Wset{(T$_k$)} $:=$ \Wset(T$_k$) $\cup$ \{X$_j$, [nv$_j$, k]\} \label{alg:sabstm_write_insert}
    }
    \nlset{17} \Else{
    
        \tcp{$X_j$ is in $\Wset(T_k)$, update current value to $v$.}
        
        \nlset{18} \Wset{(T$_k$)} $:=$ \Wset{(T$_k$)}.$\lit{update}$(X$_j$, [nv$_j$, k])\label{alg:sabstm_write_update}
    }
    \nlset{19} \Return{$ok$} \label{alg:sabstm_write_ret}
}
\end{minipage}%
\hfill
\begin{minipage}[t]{.46\textwidth}
\LinesNotNumbered
\SetNlSty{}{}{}
\SetKwFunction{TryCk}{\tryC{$_k$}}
\SetKwProg{Fn}{Fun}{:}{}
\nlset{20}\Fn{\TryCk()}{
    \nlset{21} \If{$\exists$ T$_i$ $\notin$ \cSet{(T$_k$)} : $i<k$}{\label{alg:sabstm_tryc_cset_validation}
    
        \nlset{22} \ForAll{X$_j$ $\in$ \Wset{(T$_k$)}}{

            \tcp{Update \Wset{($T_k$)} in shared memory}
            \nlset{23} \Write{(X$_j$, [nv$_j$, k])} \label{alg:sabstm_tryc_write}
        }

        \tcp{Read set validation}
        \nlset{24} \If{$\exists$ X$_j \in$ \Rset{(T$_k$)}: [ov$_j$, k] $\neq$ {read\_lvp}(T$_k$, X$_j$)}{\label{alg:sabstm_tryc_validate}

            \tcp{On validation failure, versions of each $X_j\in\Wset(T_k)$ are marked as estimated in shared memory.}
            \nlset{25} \ForAll{X$_j$ $\in$ \Wset{(T$_k$)}}{

                \nlset{26} mark\_estimate({nv}$_{j}$, k)\label{alg:sabstm_tryc_mark}
            }
            \nlset{27} \Return $A_k$\label{alg:sabstm_tryc_abort}
        }
        \nlset{28} \Return $C_k$\label{alg:sabstm_tryc_commit1}
    }
    \nlset{29} \Else{
        \tcp{$T_k$ is an independent transaction.}
        
        \nlset{30} \ForAll{$X_j \in Wset(T_k)$}{\label{alg:sabstm_tryc_wset}
        
            \nlset{31} \Write(X$_j$, [nv$_j$, k])\label{alg:sabstm_tryc_update_sm}
        }
        \tcp{Skip validation for $T_k$ and commit.}
        
        \nlset{32} \Return $C_k$\label{alg:sabstm_tryc_commit2}
    }
}

\SetKwFunction{addDependency}{$add\_dependency$}
\SetKwProg{Fn}{Fun}{:}{}
\nlset{33}\Fn{\addDependency{id, blocking\_id}}{

    \nlset{34} \If{txn\_status[blocking\_id] == Executed}{

        \nlset{35} \Return $\false$
    }
    \nlset{36} \Return $\true$
}
\end{minipage}
\end{algorithm}

\noindent\textbf{Implementation state.} For each \emph{\account{} $X_i$}, the algorithm maintains a memory location $v_i$ that stores a set of tuples $([v_1, k], [v_2, k'], \ldots)$, where each tuple $[v, k]$ represents a value $v$ written to $X_i$ by $T_k$.

\noindent\textbf{Read implementation.} The \textit{\Read{$_k(X_j)$}} function returns the most recent committed value of $X_j$ visible to transaction $T_k$. If $X_j \notin \Wset(T_k)$, the latest version written by a predecessor $T_i$ is read from shared memory at line~\ref{alg:sabstm_read_lvp} using function \textit{read\_lvp()} (as discussed in \saSupraSTM{} algorithm). If the version is an estimate and $T_i$ has not yet committed, \textit{add\_dependency} is invoked to check if $T_k$ must wait (cf. line~\ref{alg:sabstm_read_dep}); if so, $T_k$ retries at line~\ref{alg:sabstm_read_retry}. Otherwise, the version is added to the read set $\Rset(T_k)$ for subsequent validation at line~\ref{alg:sabstm_read_rset}. Otherwise, if $X_j \in \Wset(T_k)$, meaning  it is already in transaction's write set, the value is returned directly from \Wset{($T_k$)} at line~\ref{alg:sabstm_read_local}.

\noindent\textbf{Write implementation.} The \textit{\Write}$_k(X_j, v)$ performs a write that assigns a value $v$ to state $X_j$ by transaction $T_k$. It first checks whether $X_j \notin \Wset(T_k)$, if so, a new entry $[v, k]$ is inserted into $\Wset(T_k)$ (line~\ref{alg:sabstm_write_insert}); otherwise, the existing entry is updated with the new value $v$ (line~\ref{alg:sabstm_write_update}). This operation does not update in shared memory; instead, it maintains a local view of writes until commit (line~\ref{alg:sabstm_write_ret}).

\noindent\textbf{Commit implementation.} The \textit{\TryC{$_k$()}} function attempts to commit transaction $T_k$. If $T_k$ depends on any predecessor transaction(s) in \preset{}, it first updates its $\Wset(T_k)$ to shared memory (a multi-version data structure) at line~\ref{alg:sabstm_tryc_write} and then validates its read set $\Rset(T_k)$ against the latest values committed to shared memory at line~\ref{alg:sabstm_tryc_validate}. If 
validation fails, the versions written by $T_k$ in the multi-version structure are marked as estimated at line~\ref{alg:sabstm_tryc_mark}, and $T_k$ aborts by returning $A_k$ at line~\ref{alg:sabstm_tryc_abort}. If validation succeeds, or if $T_k$ is an independent transaction (its $\Wset(T_k)$ is committed to shared memory and validation is skipped), the function commits $T_k$ by returning $C_k$ at line~\ref{alg:sabstm_tryc_commit1} and line~\ref{alg:sabstm_tryc_commit2}, respectively.

}   
  {\begin{description}
    \item [Specification-Aware \SSTM{} (\sdag{})]  This approach utilizes the conflict specifications of the transactions and the \preset{} order to create a directed acyclic graph (DAG), which is used as a partial order and serves as input for the scheduler-less execution algorithm. In the DAG, transactions are represented as vertices, whereas conflicts among transactions are denoted as directed edges. The \emph{indegree} field is added with each vertex (transaction) to track dependencies with prior transactions in the \preset{} order; a transaction becomes eligible for execution when its indegree is zero. In \( C\text{set}() \), if a transaction has its independence set to have all transactions prior to it in \preset{} order, it contains no conflicts and will not have edges from any prior transactions, resulting in an indegree of zero. During execution, the non-zero indegree transactions wait for the previous transactions to commit and remove dependencies. A transaction for which execution is completed decreases the indegree of all dependent transactions, allowing them to be executed. A custom parallel-queue executor is developed to manage transaction execution and allow efficient parallel execution through multiple threads.

    \item [Specification-Aware \BSTM{} (\SHySTM{})] This approach optimizes the \BSTM{} scheduler by integrating independence information (i.e., \( C\text{set}() \)) associated with transactions, which is assumed to be provided as input. The scheduler is modified to allow transactions that are independent of all preceding transactions in the \preset{} order to execute without validation. A greater number of such independent transactions enhances the execution efficiency. This method builds upon the existing scheduler in \PEVM{}~\cite{pevm} for EVM and \BSTM{}~\cite{blockstm} for MoveVM, incorporating modifications to effectively utilize set-independence information. This approach works exactly as an optimistic speculative executor (the same as \PEVM{} and \BSTM{}) when the independence (conflict) information for all the transactions is unknown. 
\end{description}
}  


%


\section{Implementation and Evaluation of \BTM{} from Conflict Specifications}\label{Sec:experiments}
In this section, we analyze the execution latency and throughput of \saSupraSTM{} and \saBlockSTM{}, compared to baseline sequential execution and state-of-the-art parallel execution techniques.

\subsection{Implementation}
\noindent\textbf{EVM Implementation.} We compare \dBTM{} and \oBTM{} with \PEVM{}~\cite{pevm}, a version of \BlockSTM{}~\cite{blockstm} for the EVM developed by the RISE chain~\cite{riselabs}. Experiments are conducted on {REVM}~\cite{revm} version~\texttt{12.1.0}, an EVM implementation written in Rust. We used the core data structures provided by \PEVM{} version~\texttt{0.1.0} (commit hash \texttt{f0bdb21}) and implemented our algorithms on top of them. Specifically, we used a hash map to store the \cSet{}, which allowed us to store both the list of dependent transactions and the indegree of each transaction in the DAG. A custom parallel-queue is implemented to manage and enable efficient parallel execution across multiple threads.
%

We analyzed the performance on both synthetic workloads and historical blocks from Ethereum's mainnet. In the EVM, each transaction pays a gas fee to the \emph{Coinbase} account, which belongs to the block proposer. As a result, every transaction updates the {Coinbase} account, leading to a 100\% write-write conflict with all preceding transactions. To address this, we defer the fee transfer; that is, the gas fees are locally collected per transaction and credited to the {Coinbase} account only at the end, after all transactions in the block have been executed.

\noindent\textbf{MoveVM Implementation.}\label{sec:movevm}
We analyzed the performance of our \saSupraSTM{} and \saBlockSTM{} approaches against the baseline \BSTM{}~\cite{blockstm}, which runs DiemVM (an earlier version of MoveVM) using the test setup from~\cite{blockstm-setup}. The evaluation includes two primary workloads: (i) peer-to-peer (P2P) transfers, as provided in the original test setup~\cite{blockstm-setup}, and (ii) synthetic workloads (batch transfers and a generic mixed workload) we designed to mimic diverse execution patterns. 

\noindent\textbf{Experimental Setup.} We ran our experiments on a single-socket AMD machine consisting of 32~cores, 64~vCPUs, 128~GB of RAM, and running Ubuntu~18.04 with 100~GB of SSD.
The experiments are carried out in an execution setting of in-memory using REVM and MoveVM, and do not account for the latency of accessing states from persistent storage. 
We ran the experiments 52 times, with the first 2 runs designated as warm-up. Each data point in the plots represents average throughput (tps) and latency (ms), 
where each execution is repeated 50 times.

\subsection{EVM Analysis}\label{subsec:evm_analysis}
\Cref{fig:exp-syn-erc20} to~\Cref{fig:exp-e2o} show the relationship between varying threads or varying conflicts ($\alpha-$tail distribution factor) within the block, represented on the X-axis, on two important performance metrics: \emph{throughput (tps)}, represented by histograms on the primary Y-axis (Y1), and \emph{execution latency (ms)}, as line graphs on the secondary Y-axis (Y2). 

\paragraph*{Construction and Testing of Synthetic Workloads}
The experiments are performed on synthetic workloads to analyze throughput and latency under various execution conditions. A synthetic workload is a computational workload that does not depend on user input but instead simulates transaction types, concurrency patterns, or system loads. 

\input{figs/plot-EVM-synt-workloads}

We analyzed Ethereum's Mainnet blocks to generate synthetic workloads and found that certain accounts (or smart contracts) are predominantly accessed, indicating that historical blocks exhibit a tail distribution in address access frequency. Based on this observation, we designed synthetic workloads that reflect varying characteristics of transaction conflicts in blocks. Following the tail distribution, we evaluated performance across synthetic workloads, including worst-case and best-case scenarios, enabling a comprehensive assessment of system behavior under varying levels of conflict within a block. 
For this, we added another metric on the X-axis, the parameter $\alpha$, which represents the heaviness of the tail in the Pareto distribution~\cite{wikipediaPareto}. The value of $\alpha$ determines the degree of conflict, as shown in subfigure (b); smaller values of $\alpha$ correspond to fewer conflicts.

    \noindent
    \textbf{1. \ERC{} Workload (W$_{erc20}$):} This workload simulates multiple \ERC{} smart contracts, with each contract representing a distinct cluster. Transfer transactions occur among addresses within each cluster, each of which has its own \ERC{} token. 
    The number of transactions per cluster is determined using a Pareto distribution~\cite{wikipediaPareto}, while the sender and receiver addresses are selected uniformly within each cluster. This approach captures the distribution of transaction loads, resulting in certain clusters receiving higher transactions. The conflict specifications are derived using the sender, receiver, and contract addresses accessed by transactions in the block.
    
    For testing, we generate 10k \ERC{} transfer transactions across 10k contract addresses (clusters), with each cluster mapped to a unique {externally owned account (EOA)} address. Since each contract has a single EOA, transactions are effectively self-transfers. However, an EOA can initiate multiple transactions within a cluster, leading to conflicts, with a tail factor of $\alpha = 0.1$.
    
    As shown in~\Cref{fig:exp-syn-erc20}a, increasing thread count improves performance. In fact, \saSupraSTM{} achieves peak throughput at 24 threads with 411k~tps, compared to 178k~tps and 171k~tps for \saBlockSTM{} and \PEVM{}, with an average throughput of 299k, 157k and 143k~tps, respectively. In particular, \saSupraSTM{} achieves a 4$\times$ speedup over sequential and a speedup of 2.4$\times$ over \PEVM{}. 
    
    Latency trends show that optimistic approaches initially reduce execution time, but later increase it due to higher abort rates. In contrast, \saSupraSTM{} minimizes aborts by leveraging conflict specifications, reducing both abort and (re-)validation costs. Similarly, in~\Cref{fig:exp-syn-erc20}b, with 16 threads, increasing block conflicts leads to a steady performance decline, which is more pronounced in optimistic approaches due to higher abort rates and increased waiting in \saSupraSTM{}. At $\alpha = 5$ (maximum conflicts), all approaches incur a higher overhead and perform worse than sequential execution.

    \noindent
    \textbf{2. Mixed Workload (W$_m$):} This workload combines 50\% native ETH transfers and 50\% \ERC{} transfers, ensuring a balanced representation of both types of transactions. 
    For ETH transfers, sender and receiver addresses are selected using a Pareto distribution~\cite{wikipediaPareto}. The number of \ERC{} contracts is set to 25\% of the block size, with workload generation similar to \emph{W$_{erc20}$}. As a result, 2.5k contracts (clusters) comprise a block of 10k transactions. 
    
    As shown in~\Cref{fig:exp-syn-mix-vbs32}, throughput and latency remain consistent with W$_{erc20}$ workload. However, due to 50\% ETH transfers (microtransactions) in the block, all approaches show a noticeable increase in throughput. The average tps is 356k for \saSupraSTM{}, 238k for \saBlockSTM{}, 229k for \PEVM{}, and 157k for sequential. This translates to speedups of 2.27$\times$, 1.52$\times$, and 1.46$\times$ for \saSupraSTM{}, \saBlockSTM{}, and \PEVM{}, respectively, over sequential execution. The maximum tps of \saSupraSTM{} increased from 411k to 454k.
\paragraph*{Testing on Real-World Ethereum Transactions}
We selected blocks from different historical periods based on major events that may have impacted Ethereum's concurrency and network congestion. Each of these blocks allows us to analyze performance in different historical periods, providing insight into how major events, such as popular \dApp{} launches and significant protocol upgrades, affect transaction throughput and network latency. Consequently, this also helps us understand the limitations of parallel execution approaches under different network conditions, such as the gradual increase in conflicts in the block and the change in resource requirements over time. 

We trace the accessed states of transactions within a historical block using \texttt{callTracer} and \texttt{prestateTracer} APIs~\cite{chainstackTraceblock}, which provide a complete view of the block's pre-state, the state required for executing the current block. To derive conflict specifications for transactions, the pre-state file is parsed to identify all EOAs and smart contract addresses accessed within the block. These accessed addresses form the basis for constructing the conflict specification and serve as the ``ground truth''. Note that to derive specifications for a transaction T$_k$, all transactions (their pre-state EOAs and smart contract addresses) preceding it in \preset{} order are considered.

Note that real historical Ethereum blocks contain significantly fewer transactions than to our synthetic tests. To ensure a representative evaluation, we selected larger blocks from different historical periods. Specifically, we analyzed two blocks each from the CryptoKitties contract deployment and Ethereum~2.0 merge periods. We also analyzed two recent blocks with specific characteristics to demonstrate the worst-case and the best-case parallelism.

\input{figs/plot-EVM-historical-workload}

    \noindent
    \textbf{1. CryptoKitties Contract Deployment (W$_{ck}$):} The famous CryptoKitties~\cite{cryptokitties} contract was deployed on block 4605167, after which an unexpected spike in transactions caused Ethereum to experience high congestion. We analyze block 4605156, which occurred before the contract deployment, and block 4605168, which took place after the deployment. 
    
   \Cref{fig:exp-ck}a demonstrates a significant speedup in parallel transaction execution, while~\Cref{fig:exp-ck}b highlights congestion in the block following contract deployment. In \saSupraSTM{}, conflicts are over-approximated due to sender and contract address-level conflict specifications, as real blocks lack read-write and contract storage access-level granularity. This leads to a performance drop compared to synthetic tests. With complete conflict specifications and larger blocks, \saSupraSTM{} could achieve superior performance, as observed in synthetic tests.  That said throughput increases with additional threads, peaking at 8 threads, indicating the optimal point for speedup. 
    Notably, \saBlockSTM{} outperforms other approaches in this workload both before and after contract deployment. The average tps for \saBlockSTM{}, \PEVM{}, and \saSupraSTM{} drops from 69k, 68k, and 59k, before contract deployment to 54k, 53k, and 33k after contract deployment. 
    Additionally, despite smaller blocks in the pos-tcontract deployment, execution time increases for both sequential and other executors, highlighting the impact of congestion on overall performance.

    \noindent
    \textbf{2. Ethereum~2.0 Merge (W$_{e2}$):} The merge~\cite{ethereum2merge} took place in block number 15537393, which changed Ethereum's consensus to proof-of-stake along with other protocol-level changes and had an impact on transaction processing, block validation and network traffic in general. We analyze block number 15537360 before the merge and 15537421 after the merge. Compared to the W$_{ck}$ historical period, the Ethereum network had experienced an increase in block size and user activity. As shown in~\Cref{fig:exp-e2o}, the average throughput for \saBlockSTM{}, \PEVM{}, and \saSupraSTM{} before the merge was 55k (2.68$\times$ speedup over sequential), 56k (2.73$\times$ speedup), and 33k (1.63$\times$ speedup), respectively. Post-merge, these values changed to 56k (3.17$\times$ speedup), 53k (2.98$\times$ speedup), and 23k (1.29$\times$ speedup). Notably, the increased speedup over sequential execution post-merge shows the growing parallel execution potential to improve throughput.

    \noindent
    \textbf{3. Worst-case and Best-case Blocks (W$_{wb}$):} To analyze performance extremes, we selected Ethereum blocks exhibiting different conflict profiles. The worst-case block, 17873752, represents maximal conflicts with densely interconnected transactions and consists of 1189 transactions, of which 1129 are ETH transfers, 21 are \ERC{} transfers, and 39 are other contract transactions. We observed that most ETH transfers interact with the same address, resulting in a longer dependency chain of length 1104. As a result, the block becomes inherently sequential, and hence a good choice for worst-case analysis. In contrast, the best-case block, 17873654, consists predominantly of independent transactions. There are a total of 136 transactions, including 41 ETH transfers, 18 \ERC{} transfers, and 77 other contract transactions. We chose this block because it has a significant number of internal transactions (435) and the length of the longest conflict dependency chain is 50.

    As shown in \Cref{fig:exp-wb_case}a, in the worst-case, for \saSupraSTM{} the tps remains consistently low, around 100k-137k across all thread counts and 166k-195k for \saBlockSTM{}. This reflects minimal improvement over the sequential tps of 162k, highlighting the limitation of parallel execution under blocks with high access-level dependencies. While \PEVM{} at 4 threads shows 1.41$\times$ performance gains over sequential, increasing the threads beyond introduces aborts and re-executions, resulting in overhead. However, the performance gain could be due to the lazy update of ETH transfers, which defers balance changes to EOAs until a later stage at commit time, as well as to independent contract transactions that are computationally intensive compared to ETH transfers$-$microtransactions with minimal computational logic. Due to the overhead of DAG construction, \saSupraSTM{} shows a noticeable performance degradation. 
    Moreover, 94.95\% of the transactions in the block are ETH transfers resulting in smaller execution latency despite the large block size. The average execution latency for \PEVM{} is 6.36~ms, \saBlockSTM{} is 6.45~ms, that are not significantly less compared to 7.36~ms of sequential execution, while \saSupraSTM{} incurs slightly higher latency 9.97~ms.

    In contrast, the best-case block in \Cref{fig:exp-wb_case}b shows that \saBlockSTM{} achieves throughput in the range of 3.9k-6.8k~tps, while \saSupraSTM{} ranges from 6.9k-7.2k~tps, both significantly outperforming 3.74k~tps of the sequential execution. The average execution latency across threads is 27.1~ms for \saBlockSTM{} and 19.2~ms for \saSupraSTM{}, compared to 36.36~ms for sequential and 27.25~ms for \PEVM{}. In terms of average speedup, \saBlockSTM{} achieves a 1.34$\times$ improvement over sequential and performs on par with \PEVM{}, while \saSupraSTM{} shows superior efficiency, with an average speedup of 1.89$\times$ over sequential and a 1.42$\times$ speedup over \PEVM{}. Since, this block contains a higher proportion of smart contract transactions with numerous internal calls, making it computationally intensive for the VM. Consequently, minimizing aborts and re-executions yields tangible performance benefits. Notably, \saSupraSTM{} maintains its competitive edge due to its reduced abort and re-execution overheads.

Observe that from one historical period to another (even with just these blocks), despite the block size increasing, the overall throughput has decreased, highlighting the impact of congestion and resource requirements on overall performance. 

\noindent\emph{\textbf{Summary.} Our EVM analysis highlights the strong performance of \saSupraSTM{} and \saBlockSTM{}, demonstrating their potential to achieve a near-theoretical maximum parallel execution by leveraging conflict specifications. With more accurate conflict specifications in historical workloads, we can expect \saSupraSTM{} and \saBlockSTM{} to achieve even higher speedups over sequential and \PEVM{}. However, in the absence of conflict specifications, \saBlockSTM{} is expected to perform the same as the optimistic execution of \PEVM{}. On average, across synthetic workloads, \saSupraSTM{} achieves a maximum of $4\times$ speedup over sequential and $3.82\times$ over \PEVM{}, while \saBlockSTM{} achieves $1.8\times$ and $1.2\times$ speedups, respectively. In contrast, across historical blocks, \saSupraSTM{} achieves an average speedup of $1.22\times$ over sequential and a maximum of $1.7\times$ over \PEVM{}, while \saBlockSTM{} achieves an average of $1.83\times$ over sequential and a maximum of $1.41\times$ over \PEVM{}.}

\subsection{MoveVM Analysis}\label{subsec:movevm_analysis}
\revision{
We evaluate and compare the performance of our \saSupraSTM{} and \saBlockSTM{} against the baseline \BSTM{}~\cite{blockstm} running DiemVM (a specific version of MoveVM) for execution in test-setup at~\cite{blockstm-setup}. The evaluation includes two primary workloads: a peer-to-peer (P2P) transfer workload, as available in the original test setup~\cite{blockstm-setup}, and a set of synthetic smart contract workloads designed to mimic different execution patterns. 

\input{figs/plot-MoveVM-synt-p2p}

\noindent
\textbf{MoveVM Analysis.} 
To demonstrate this, \Cref{fig:exp-movevm-p2p} describes our results for the implementation of \saSupraSTM{} and \saBlockSTM{} on the MoveVM when compared against the state-of-the-art \BSTM{} and sequential execution. We perform these experiments on the \BSTM{} benchmark test-setup~\cite{blockstm-setup}, which runs a virtual machine for smart contracts in the Move language~\cite{blackshear2019move}. Performance (throughput as histogram on the Y1-axis and latency as line chart on the Y2-axis) is tested on P2P transactions, which essentially translate to smart contract transactions in the Move ecosystem. 

\noindent
\textbf{1. P2P Transfer Workload (W$_{p2p})$:}
In a block of 10k transactions, every P2P transaction randomly chooses two different accounts and transfers a fixed amount between them. The access specifications are derived by analyzing the sender and receiver of each transaction with all transactions preceding it in \preset{} order. The conflicts in the block are determined by the number of accounts; specifically, when there are only two accounts, the load is inherently sequential (each transaction depends on the one prior to it). Performance is tested on varying threads (X-axis) on two different account setups: 100 (high contention,~\Cref{fig:exp-movevm-p2p}a) and 10k (less contention,~\Cref{fig:exp-movevm-p2p}b). We also tested performance on varying accounts (X-axis) with fixed threads (to 16 and 32), as shown in~\Cref{fig:exp-movevm-p2p}c and~\Cref{fig:exp-movevm-p2p}d. 

As shown in~\Cref{fig:exp-movevm-p2p}, the results reveal that \saSupraSTM{} and \saBlockSTM{} have throughput that is comparable to \BSTM{} and are significantly better with 16 and 32 threads with 100 accounts for \saSupraSTM{} in~\Cref{fig:exp-movevm-p2p}a. As shown, \saSupraSTM{} attains a maximum throughput of 54k at 16 threads, while \saBlockSTM{} and \BSTM{} attain a maximum throughput of 46k at 32 threads. The average tps is 34.8k, 29.99k and 29.89k for \saSupraSTM{}, \saBlockSTM{} and \BSTM{}, respectively. 

Similarly, in the other three experiments, \saSupraSTM{} outperforms other approaches, demonstrating the advantageous effects of knowing the conflict specification in advance. Note that with just two accounts for P2P transactions, it will be a 100\% conflict case with a linear chain of conflicts, as shown in~\Cref{fig:exp-movevm-p2p}c and~\Cref{fig:exp-movevm-p2p}d. In this case, due to overhead, sequential execution outperforms parallel execution approaches. However, the tps increases linearly with accounts, until all transactions become independent, as can be seen that from 1k to 10k, the tps is almost saturated for both 16 threads (\Cref{fig:exp-movevm-p2p}c) and 32 threads (\Cref{fig:exp-movevm-p2p}d). 

\input{figs/plot-MoveVM-synt-workloads}

\noindent
\textbf{2. Batch Transfer (W$_{batch}$) and Generic Workload (W$_{gw}$):}
To evaluate performance beyond P2P transactions, we designed a batch transfer workload in which a single account transfers funds to multiple other accounts. In this workload, to follow the account access distribution in real-world workloads, we used a tail distribution, similar to the synthetic workload of the EVM analysis in~\Cref{fig:exp-syn-mix-vbs32}. For this, for a block size of 10k and 10k accounts, when varying threads (X-axis), we tested performance under a fixed $\alpha = 0.1$ (\Cref{fig:exp-syn-move-batch-trans}a), as well as we evaluated performance on 16-thread count with varying $\alpha$ on X-axis, as shown in \Cref{fig:exp-syn-move-batch-trans}b. Moreover, we also tested performance on a generic workload (\Cref{fig:exp-syn-movevm-generic}) in the same setting, where the block consists of an equal number of P2P and batch transfers.

As shown in \Cref{fig:exp-syn-move-batch-trans}, in low-conflict settings ($\alpha = 0.1$), both \BlockSTM{} and \saBlockSTM{} marginally outperform \saSupraSTM{}. Among them, \saBlockSTM{} achieves the highest throughput of 49k~tps at 32 threads and begins to outperform \BlockSTM{} from 8 threads onward. However, as $\alpha$ increases and conflicts become more frequent, the performance of all approaches degrades. In these high-conflict settings, \saSupraSTM{} begins to outperform other approaches from $\alpha \geq 0.25$, though its maximum throughput at $\alpha = 0.1$ remains 33k.

In the case of the generic workload, as shown in~\Cref{fig:exp-syn-movevm-generic}, \saSupraSTM{} initially performs better than \BlockSTM{}, but as the thread count increases, \BlockSTM{} and \saBlockSTM{} start outperforming, with \saBlockSTM{} showing a slight edge. Upon increasing $\alpha$ further, the performance of all approaches drops and \saSupraSTM{} outperforms the others beyond $\alpha = 0.25$. At $\alpha = 5.0$, all approaches perform close to the sequential baseline. 
To improve on the worst case and to get more balanced performance in general (maximum throughput in all cases), we believe a workload-adaptive parallel execution approach is the best path forward (that is briefly discussed in~\Cref{sec:conc}).

\noindent\emph{\textbf{Summary.} Our MoveVM analysis highlights the strong performance of \saSupraSTM{} and \saBlockSTM{}. It shows that conflict specifications in conjunction with optimistic execution via \saBlockSTM{}, where the scheduler is aware of transactional dependencies in advance minimize the overhead of aborts, resulting in improved performance relative to optimistic execution of \BlockSTM{}.}

}
\section{\BTMemory{} for EVM}\label{sec:spec}
%
Having established the empirical advantages of \BTM{'s} leveraging conflict specifications, we now present the design of a \emph{full} \emph{conflict analyzer} integrated implementation of the EVM parallel execution. 
We demonstrate with the conflict analyzer that it is possible to algorithmically implement \emph{sound} conflict specifications for real Ethereum blocks (cf. \Cref{subsec:conf_spec}). The derived conflict specifications, though they might be \emph{incomplete}, i.e., they may not identify conflict relations for some class of transactions, can be generated efficiently and thus prove consequential for maximizing transaction execution throughput (as established in~\cref{sec:algo_overview}). 
Additionally, we present results for a 
workload adaptive execution leveraging the conflict specifications 
for EVM in~\Cref{subsec:adaptive}. 

%
\label{spec:formalism}
\subsection{Conflict Specifications for EVM}\label{subsec:conf_spec}
\noindent\textbf{Preliminaries.}
We now detail some technical preliminaries needed to describe the approach for building conflict specifications in EVM.
%
\begin{definition}
Ethereum blockchain is made up of a set of \textbf{accounts} $A = A_U \cup A_C$, where $A_U$ is the set of EOAs and $A_C$ is the set of contract accounts. Each account $a \in A$ is associated with some data $a.\texttt{balance}$ such that $a.\texttt{balance} \in a.\texttt{data}$ where $a.\texttt{balance}$ is the ETH balance. For an EOA $a \in A_U$, $a.\texttt{data} = \{a.\texttt{balance}\}$ whereas a contract account usually has fields associated in $a.\texttt{data}$.
\end{definition}
\begin{definition}
An Ethereum transaction is specified by a tuple $(\texttt{origin}, \texttt{dest}, \texttt{value}, \texttt{calldata})$ where $\texttt{origin} \in A_U$ is the initiator of the transaction, $\texttt{dest} \in A$ is the transaction recipient, $\texttt{value} \in \mathbb{N}$ is the number of wei sent by \texttt{origin} to \texttt{dest}, $\texttt{calldata}$ is the accompanying data. An Ethereum transaction $T$ is a simple ETH payment if $T.\texttt{dest} \in A_U$ and is a contract transaction otherwise. For the read/write sets it always holds that \Rset(T), \Wset(T) $\subseteq \cup_{a \in A}a.\texttt{data}$.
\end{definition}

\begin{note}
A real-world Ethereum transaction contains more elements, but for the purpose of overapproximating read/write sets here, such tuple is a sufficient enough characterization.
\end{note}

\begin{definition}
For a contract transaction $T$, the \textbf{signature} of $T$ denoted by $T.\texttt{sig}$ is the first four bytes of $T.\texttt{calldata}$ that specify the function to call in the contract ($T.\texttt{dest}$).
\end{definition}
\begin{example}
Consider a situation where an EOA $a \in A_U$ initiates a transaction $T_1$ calling the function \texttt{receiveMessage} of a \emph{Contract b} (cf. \cref{lst:cap0}) so the corresponding tuple is: \((\texttt{origin}: a, \texttt{dest}: b, \texttt{value}: 1, \texttt{calldata}: \{\texttt{signature}: \texttt{receiveMessage}, \texttt{arg}: ``hello"\})\).
We can see here, that always $\Rset(T_1) = \{a.\texttt{balance}, b.\texttt{balance}, b.\texttt{shouldAccept}\}$. Notice that whenever $T.\texttt{value}>0$ we do have $T.\texttt{origin}.\texttt{balance}, T.\texttt{dest}.\texttt{balance} \in \Rset(T),\Wset(T)$ as the value needs to be deducted from $T.\texttt{origin}$ and added to $T.\texttt{dest}$.  But either $\Wset(T_1) = \emptyset$ or $\Wset(T_2) = \{b.\texttt{message}\}$ depending on the value of $b.\texttt{shouldAccept}$, because the message field is updated only if the \texttt{shouldAccept} value is true.

Now, let another EOA account $c$ initiate a transaction $T_2$ with the corresponding tuple \((\texttt{origin}: c, \texttt{dest}: b, \texttt{value}: 0, \texttt{calldata}: \{\texttt{signature}: \texttt{setShouldAccept}, \texttt{arg}: \texttt{false}\})\). 
We have $\Rset(T_2) = \emptyset, \Wset(T_2) = \{b.\texttt{shouldAccept}\}$. We can see from the read and write sets characterized here that $T_1,T_2$ can never have a read-from conflict if $T_1$ precedes $T_2$ (regardless of the control flow in $T_1$) in the \preset{} order. But if $T_2$ precedes $T_1$, then they will always have read-from conflict since $T_2$ always writes to $b.\texttt{shouldAccept}$ and $T_1$ always reads from it.
\end{example}
\noindent\textbf{Approach.} The EVM conflict analyzer is just labeling simple payments as $SimplePayment$, the designated \ERC{} functions with their respective signature and everything else as $exitsContract$ in line~\ref{spec:algo:weak-label}.
After the labeling, 
check if a transaction $T_j$ can have read-from conflicts with some $T_i, i<j$, we consider multiple scenarios. If $T_j$ is labeled $exitsContract$, we simply assume that they do have read-from conflict (line~\ref{spec:weak-algo:exit}). If both $T_i,T_j$ are labeled as one of the special \ERC{} transactions and interact with the same contract, it just checks to see if the sender and accounts in the arguments of the functions called for $T_j$ are disjoint from those of $T_i$ (line~\ref{spec:algo:weak-token}). This is enough as an \ERC{} contract only affects some mapping ($balances,allowances$) in the contract storage with the keys affected being either the transaction arguments or the initiator ($T.origin$). 

In all other cases, it is enough to check that the origin and destination of $T_i$ is disjoint from those of $T_j$ (line~\ref{spec:algo:weak-distinct}). Since both \ERC{} transactions and simple payments only affect the source and target of the transactions, this is enough to guarantee independence.

\begin{lstlisting}[caption={Smart Contracts}\label{lst:cap0}]
contract b {
    string message = "";
    bool shouldAccept = true;
    
    function receiveMessage(String message) external payable{
        if(shouldAccept) {
            this.message = message;
        }
    }
    
    function setShouldAccept(bool shouldAccept) external {
        this.shouldAccept = shouldAccept;
    }
}


contract Token {
    mapping(address => uint) balances;
    address priceOracle;
    
    function transfer(address target, uint amount) external {
        require( tokenBalances[msg.sender] >= amount );
        balances[msg.sender] -= amount;
        balances[target] += amount;
    }
    
    function turnEtherToToken() external payable {
        (, bytes memory data) = priceOracle.staticcall( abi.encodeWithSignature("tokenPrice()") );
        uint tokenPrice = bytesToInt(data);
        balances[msg.sender] = balances[msg.sender] + (msg.value/tokenPrice);
    }
}


contract Wallet {
    mapping(address => uint) balances;
    
    function addToWallet() external payable {
        balances[msg.sender] += msg.value;
    }
    
    function withdraw(uint amount) external {
        require(balances[msg.sender] >= amount);
        balances[msg.sender] -= amount;
        msg.sender.transfer(amount);
    }
}
\end{lstlisting}

\begin{algorithm}[!b]
    \scriptsize
    \caption{Output conflict specifications that is given as input to \Cref{algo:sdag,algo:saBlockSTM}.
    }
    \label{spec:weak-algo_main}
    \SetKw{Continue}{continue;}
    \KwIn {A transaction $T_j$ in a \preset{} order $T_1\to,\ldots,\to T_n$.} 
    \KwOut {$S \subseteq \{T_i | i<j,\Rset(T_j) \cap \Wset(T_i) = \emptyset\}$.} 
    \begin{minipage}[t]{1\textwidth}
        \LinesNotNumbered
        \SetNlSty{}{}{}
        \SetKwFunction{isToken}{is\ERC{}}
        \SetKwProg{Fn}{Fun}{:}{}
        \Fn{\isToken($label$)}{
            \nlset{2} \Return $label \in \{transfer,transferFrom,approve\}$
        }

        \SetKwFunction{ercAccounts}{erc20Accounts}
        \SetKwProg{Fn}{Fun}{:}{}
        \nlset{3}
        \Fn{\ercAccounts($T$)}{\label{alg:alg2l34}
        \tcp{Get the affected accounts in an \ERC{} transaction, for instance, a \emph{transferFrom()} involves \emph{from}, \emph{to}, and the \emph{sender} address. Two \ERC{} transfers are independent if their affected account sets are disjoint.}
        
            \nlset{4} \If{$label \in \{transfer,approve\}$}{
                \nlset{5} \Return $\{T.\texttt{origin}, T.\texttt{calldata}.\texttt{to}\}$
            }
            \nlset{6} \ElseIf{$label = transferFrom$}{
            
                \nlset{7} \Return $\{T.\texttt{origin}, T.\texttt{calldata}.\texttt{from}, T.\texttt{calldata}.\texttt{to}\}$
            }
        }
        
        \SetKwFunction{weakLabel}{getLabel}
        \SetKwProg{Fn}{Fun}{:}{}
        \nlset{8}\label{spec:algo:weak-label}\Fn{\weakLabel($T$)}{
                \nlset{9} \If{$T.dest \in A_U$}{
                \nlset{10} \Return $SimplePayment$.
            }
            \nlset{11} $sig \gets T.\texttt{sig}$
        
            \nlset{12} 
            \uIf{$sig \in \{transfer,transferFrom,approve\}$ in \ERC{}}{ 
                \nlset{13} \Return $T.\texttt{sig}$
            }
            \nlset{14} \Else{
                \nlset{15} \Return $exitsContract$
            }
            
        }

    \SetKwFunction{rfconset}{findCSet}
    \SetKwProg{Fn}{Fun}{:}{}
    \nlset{16}
    \Fn{\rfconset{$T_j$}}{
        \nlset{17} $label_j \gets \weakLabel(T_j)$
        
        \nlset{18} \If{$label_j = exitsContract$}{
        \label{spec:weak-algo:exit}
            \nlset{19} \Return $\emptyset$
        }
        
        \nlset{20} $Output \gets \emptyset$
        
        \nlset{21} \ForAll{$1 \leq i < j$}{
            \nlset{22} $label_i \gets \weakLabel(T_i)$.
            
            \nlset{23} \If{$T_i.\texttt{origin} = T_j.\texttt{origin}$ or $label_i = exitsContract$}{
                \nlset{24} \Continue
            }
            
            \nlset{25}\uIf{$T_i.dest = T_j.dest$ and \isToken($label_i$) and \isToken($label_j$)}{\label{spec:algo:weak-token}
            
                \nlset{26} $affectedAccounts_i$ $\gets$ \ercAccounts($T_i$)
                
                \nlset{27} $affectedAccounts_j$ $\gets$ \ercAccounts($T_j$)
    
                \nlset{28} \If{$affectedAccounts_i \cap affectedAccounts_j = \emptyset$}{
                
                    \nlset{29} $Output \gets Output \cup \{T_i\}$
                }
            }
            \nlset{30} \label{spec:algo:weak-distinct} \ElseIf{$|\{T_i.\texttt{origin}, T_j.\texttt{origin}, T_i.\texttt{dest}, T_j.\texttt{dest}\}| = 4$}{ 
            
                \nlset{31} $Output \gets Output \cup \{T_i\}$
            }
        }
        
        \nlset{32} \Return $Output$\label{alg:alg2l66}
    }
    \end{minipage}
\end{algorithm}

We now go over the \Cref{spec:weak-algo_main} for transactions interacting with \emph{Token contract} and \emph{Wallet contract}, as shown in \Cref{lst:cap0}. Notice that every transaction that calls a function other than $Token.transfer$ in these two contracts will be labeled as $exitsContract$ and thus be assumed to conflict with everything else. Therefore, just consider an account $a_1$ initiating a transaction $T_1$ invoking $Token.transfer$ with address $a_2$ and another account $b_1$ initiating $T_2$ invoking $Token.transfer$ with address $a_2$. Here, both functions are labeled $transfer$. Now, to check whether they are independent, one needs to check whether all $a_1,a_2,b_1,b_2$ are distinct.
Now, let $T_3$ be a simple payment of $c_1$ to $c_3$. Of course $T_3$ is labeled as $SimplePayment$. Similarly, one only needs to check the distinctness of $c_1, c_2$ with respect to each of $a_1,Token$ and $b_1,Token$ to see if any conflict arises.

\begin{proposition} 
\label{case:wellformed}
    \begin{itemize}
        \item Whether $T$ is a simple payment or calls a designated \ERC{} function, only writes to or reads from accounts $T.\texttt{origin},T.\texttt{dest}$.
        
        \item If $T$ calls an \ERC{} function in $transfer,transferFrom,approve$, then only $T.\texttt{dest}$ storage is affected. Assuming that the contract adheres to \ERC{} specifications, as long as $T.\texttt{origin}$ and accounts in $T.\texttt{data}$ are disjoint from another $T'$ calling these specific functions in $T.\texttt{dest}$, then different parts of the same mappings in $T.\texttt{dest}$ are modified. Also, if the contract is implemented as a proxy, we assume that only the proxy contract can call the parent (no EOA can initiate a transaction directly with the parent contract).
    \end{itemize}
\end{proposition}

\begin{theorem}\label{th:cfweak} 
    If the \cref{spec:weak-algo_main} outputs an independence set $\mathcal{T}$ for a transaction $T_j$, then for each $T \in \mathcal{T}$ in the \preset{} order before $T_j$, $\Rset(T_j) \cap \Wset(T) = \emptyset$.
\end{theorem}
\begin{proof}
Let $T \in \mathcal{T}$ in the \preset{} order be before $T_j$. First note that by the algorithm, if $T$ label is $exitsContract$, then it is overapproximated to have read-from conflict with every transaction before in the \preset{} order. So we can assume $T$ is either a simple payment or an \ERC{} transaction. Then (according to the definition of simple payment and \cref{case:wellformed} and line~\ref{spec:algo:weak-label}) $T$ only modifies/reads from $\{T.origin,T.dest\}$ which is guaranteed to be disjoint from $\{T_j.\texttt{origin}, T_j.\texttt{dest}\}$ by line~\ref{spec:algo:weak-distinct}. Therefore, by line~\ref{spec:algo:weak-label} and \cref{case:wellformed}, $T_j$ only modifies $\{T_j.\texttt{origin},T_j.\texttt{dest}\}$ which is disjoint from the read set of $T$ implying that $T$ does not have a read-from conflict with $T_j$.
\end{proof}
%

\subsection{Conflict Analyzer Integration with EVM for Parallel Execution}\label{subsec:conf_analyzer_evm_exp}
This section presents the analysis of \emph{full} conflict analyzer integrated with our \BTM{} algorithm for the EVM. 
We evaluated the performance of the proposed \BTM{} algorithm on the EVM, leveraging the conflict specification provided by the conflict analyzer on both synthetic and historical blocks. 

\noindent\textbf{Implementation Details.}
We introduced several implementation-level modifications to the \oBTM{}, resulting in an {extended} version that we refer to as \iBTM{}. {It uses analyzer-provided dependency information (beyond the sender, receiver, and contract addresses used in \oBTM{}) to reduce aborts, and applies the lazy update optimization (of \PEVM{}) for native ETH transfers, while other transactions remain optimistically executed.} We choose to perform tests with \osaBlockSTM{}, as it required minimal changes to existing implementation and demonstrates the potential benefits of the conflict analyzer when integrated with the optimistic execution of \saBlockSTM{}. Note that for this analysis, we re-run the experiments in our integrated conflict analyzer setup, and the numbers in~\Cref{subsec:evm_analysis} can be treated as the ground truth. 
The entire implementation built on top of \PEVM{} which is 8097 lines of code with 4454 lines added over the original \PEVM{} implementation.


The conflict analyzer is implemented in approximately 700 lines of Rust code. It uses HashMap that maps each accessed address to a vector of transactions that have interacted with it. 
Transactions labeled as $exitsContract$ are pruned in the first pass to reduce the number of pairs evaluated for independence as a transaction labeled $exitsContract$ is assumed to conflict with all others. After that if the number of transactions remaining after that is high enough (a manually set parameter), parallelization is done across 16 threads.

{In addition to the implemented version of the conflict analyzer detailed in \Cref{subsec:conf_spec}, we additionally implemented another \emph{strong} version of the conflict analyzer. \revision{The strong version of the conflict analyzer is presented in \Cref{sec:conf_spec_evm}.} Both of these are automated completely and proved to be sound (but not complete) if correctly implemented (cf. \Cref{th:cfweak}). The implemented analyzer main limitation is its inability to deal with non \ERC{} contracts, marking every transaction to such contracts as dependent on all other transactions. Strong analyzer main limitation is that when a contract function calls another contract $C$ not known statically (that is, the address of the $C$ is determined by the input data of the transaction), one cannot know what sequence of blockchain state manipulations after that and thus such transactions are marked dependent on everything else. In the case that $C$ and its payload are known statically, it is rather easy to add more rigorous specifications.
There exist inputs which can make the analyzer derive no independence specification making the execution as slow as sequential. To be more precise, we analyze the \ERC{} contracts by their function signature, so a malicious contract not adhering to \ERC{} specifications may result in an unsound analysis, however, one can regularly inspect deployed contracts either manually or automatically to adhere to specifications and analyze only those marked to adhere to specifications (which can be done by current tools \cite{etherscan,MythX}). 
}



\begin{table}[!tb]
    \caption{Conflict analyzer integrated with \osaBlockSTM{} for parallel execution in EVM; metrics from conflict analysis phase, block composition, and execution performance across different approaches.}
    \label{tab:conflict-analysis_wc_et}
    \resizebox{\textwidth}{!}{%
    \begin{tabular}{|r|r|c|c|c|c|c|c|}
        \hline
            & & \multicolumn{4}{c|}{\textbf{Ethereum Historical Blocks}} & \multicolumn{2}{c|}{\textbf{Synthetic 10k Blocks}}\\
            \cline{3-8}
            & \multirow{2}{*}{\shortstack{\textbf{Metric}}} & 
              \multirow{2}{*}{\shortstack{\textbf{4605168}\\\textbf{(CryptoKitties)}}} & 
              \multirow{2}{*}{\shortstack{\textbf{15537421}\\\textbf{(Ethereum~2.0)}}} & 
              \multirow{2}{*}{\shortstack{\textbf{17873752}\\\textbf{(Worst-case)}}} & 
              \multirow{2}{*}{\shortstack{\textbf{17873654}\\\textbf{(Best-case)}}} & 
              \multirow{2}{*}{\shortstack{\textbf{\ERC{}}}} & 
              \multirow{2}{*}{\shortstack{\textbf{Mix}}} \\
              & & & & & & &\\
        \hline\hline

        \multirow{4}{*}{\centering\shortstack{\textbf{Conflict}\\\textbf{Analysis}\\\textbf{Phase}}} 
            & & & & & & &\\
            & \textbf{Conflict Generation Time} 
            & 18.314 µs 
            & 39.00  µs  
            & 106.79 µs 
            & 30.09  µs  
            & 6.06   ms
            & 3.30   ms
            \\
            
            & \textbf{Specification Size}            
            & 627      
            & 943       
            & 50706     
            & 1691  
            & 49658226
            & 49824316
            \\

        & & & & & & &\\
        \hline
        
        & & & & & & &\\
        \multirow{5}{*}{\centering\shortstack{\textbf{Block}\\\textbf{Stats}}}
            & \textbf{Block Size}                    
            & 77       
            & 173       
            & 1189      
            & 136   
            & 10000
            & 10000
            \\
            & \textbf{ETH Transactions}       
            & 32       
            & 33        
            & 1129      
            & 41  
            & 0
            & 5000
            \\
            & \textbf{\ERC{} Transactions}           
            & 9        
            & 11        
            & 21        
            & 18   
            & 10000
            & 5000
            \\
            
            & \textbf{Other Contract Transactions} 
            & 36
            & 129 
            & 39
            & 77
            & 0
            & 0
            \\
            
            & \textbf{Dependency Count} 
            & 10 
            & 7  
            & 4  
            & 0  
            & 5683
            & 3215
            \\
        & & & & & & &\\
            
        \hline
        
        & & & & & & &\\
        \multirow{4}{*}{\shortstack{\textbf{Average}\\\textbf{Execution}\\\textbf{Time}\\\textbf{(\textbf{16 threads)}}}} 
            & \textbf{\textcolor{black!80}{Sequential}} 
            & \textcolor{black!80}{2.107 ms}
            & \textcolor{black!80}{15.60 ms} 
            & \textcolor{black!80}{7.21 ms}
            & \textcolor{black!80}{36.63 ms}
            & \textcolor{black!80}{114.43 ms}
            & \textcolor{black!80}{72.12 ms}
            \\

            & \textbf{\textcolor{teal!80}{\PEVM{}}} 
            & \textcolor{teal!80}{1.32 ms} 
            & \textcolor{teal!80}{6.93 ms} 
            & \textcolor{teal!80}{6.03 ms} 
            & \textcolor{teal!80}{27.52 ms} 
            & \textcolor{teal!80}{96.75	ms}
            & \textcolor{teal!80}{53.60 ms}
            \\

            & \textbf{\textcolor{purple!80}{\oBTM{}}} 
            & \textcolor{purple!80}{1.24 ms} 
            & \textcolor{purple!80}{6.88 ms} 
            & \textcolor{purple!80}{5.99 ms} 
            & \textcolor{purple!80}{27.68 ms} 
            & \textcolor{purple!80}{88.67 ms}
            & \textcolor{purple!80}{54.14 ms}
            \\
            
            & \textbf{\textcolor{orange!80}{\iBTM{}}} 
            & \textbf{\textcolor{orange!80}{1.23 ms}}
            & \textbf{\textcolor{orange!80}{6.43 ms}}
            & \textbf{\textcolor{orange!80}{5.43 ms}}
            & \textbf{\textcolor{orange!80}{26.85 ms}}
            & \textbf{\textcolor{orange!80}{79.42 ms}}
            & \textbf{\textcolor{orange!80}{39.50 ms}}
            \\
            
        & & & & & & &\\
        
        \hline
            
    \end{tabular}%
}
\end{table}
\input{figs/plot-EVM-integrated}

\noindent
\textbf{Testing on Synthetic Workloads.}
\Cref{tab:conflict-analysis_wc_et} presents comprehensive details of integrated experiments, including conflict analysis and execution latency on historical Ethereum and large synthetic blocks. The reported metrics capture statistics from the conflict analysis phase, including the time required to generate conflict specifications and their size, as well as block characteristics such as block size, transaction type distribution, and conflict dependency counts for block. Also report the execution time under three configurations: sequential, \PEVM{} and \osaBlockSTM{} (execution with 16 threads). The goal is to evaluate computational cost and execution efficiency in both historical and large synthetic blocks.

For analysis, we choose two synthetic workloads from Ethereum analysis \Cref{subsec:evm_analysis}: \ERC{} and Mix workloads, with a block size of 10k transactions and a fixed conflict distribution factor $\alpha = 0.1$, while varying the thread count on the X-axis. \Cref{fig:exp-cewb_intgrated_evm_new} shows the advantage of the conflict analyzer in saving the abort and re-execution cost for optimistic execution. As the thread count increases, in both workloads, \ERC{} transfer (\Cref{fig:exp-cewb_intgrated_evm_new}a) and Mix (\Cref{fig:exp-cewb_intgrated_evm_new}b) a noticeable improvement in execution latency can be observed for \osaBlockSTM{} over \PEVM{}, highlighting the benefits of leveraging conflict specifications for execution. The benefit is especially pronounced for the \ERC{} workload, which contains more uniform transaction patterns.

\noindent
\textbf{Testing on Historical Ethereum Blocks.} 
{Figures \ref{fig:exp-cewb_intgrated_evm_new}c$-$f} illustrates the performance on Ethereum blocks. We select one block each from the CryptoKitties (same block as in~\Cref{fig:exp-ck}b) and the Ethereum~2.0 workloads (same block as in~\Cref{fig:exp-e2o}b). We also analyzed performance on the worst and best-case blocks of the workload in~\Cref{fig:exp-wb_case}. As shown in \Cref{fig:exp-cewb_intgrated_evm_new}, with our conflict analyzer, read-write conflict detection improves the performance of \iBTM{}, which can be clearly observed by comparing the relative performance of \PEVM{} with \iBTM{}. The performance of \iBTM{} improves in almost all cases compared to \PEVM{} and \oBTM{} with a noticeable difference, which is expected to improve further with a fully optimized and mature conflict analyzer. This performance could be due to the reduced number of aborts and re-executions (re-validations) in the optimistic execution.

\noindent\emph{\textbf{To conclude,} the proposed \iBTM{} achieves up to a $1.75\times$ speedup over sequential execution and a $1.33\times$ speedup over \PEVM{}, with an average speedup of $1.24\times$ over \PEVM{} in synthetic workloads. In historical workloads, it achieves a maximum speedup of $2.43\times$ and $1.14\times$ over sequential and \PEVM{}, also consistently outperforms \PEVM{} across all evaluated scenarios.}

\subsection{Adaptive Implementation Based on Conflict Threshold for EVM}\label{subsec:adaptive}
The advantage of the conflict specification is that it allows us to determine how many pairwise conflicts exist in a block of transactions. We demonstrate how we can leverage this for an adaptive implementation that can also handle high-conflict workloads (\Cref{fig:exp-wb_case_adaptive}a). Specifically, by computing a \emph{conflict threshold} for each block with minimal overhead, we deterministically fall back to sequential execution in the case of high conflicts.

\Cref{fig:exp-wb_case_adaptive} demonstrate the advantage of an adaptive technique that dynamically selects the most suitable execution path based on the characteristics of the workload. When $\alpha = 0.1$, the \iBTM{} outperforms all other approaches, leading the adaptive technique to choose it for execution. In contrast, when $\alpha = 5$, the workload becomes highly sequential, and all parallel approaches perform worse than sequential execution. The adaptive approach could fall back to sequential execution; however, there could be small overhead with selecting the optimal approach; this is offset by significant execution time savings in highly conflicting workloads. The current adaptive mechanism is implemented and evaluated for two execution paths, sequential or \iBTM{}, and considers the conflict threshold and block size for decision making. Exploring additional metrics beyond the conflict threshold and block size, such as available compute, expected block computation cost, and past $n$ block execution statistics, as well as different adaptive parallel execution paths (e.g. \dBTM{}, \iBTM{}, \PEVM{}, sequential), remains ongoing work and may constitute standalone future work.

\input{figs/plot-EVM-synt-adaptive}

\section{\BTMemory{} for MoveVM}\label{subsec:movevm_btm_integrated}
\revision{In \Cref{sec:aptos-model}, we demonstrate how we derive the conflict specifications using conflict analyzer for Aptos's MoveVM. We then present a limited set of performance numbers for conflict analyzer integrated \BTM{} on the MoveVM in~\Cref{subsec:conf_analyzer_movevm_exp}. The \BTM{} implementation for MoveVM leverages a custom-built conflict analyzer that is more conservative than the EVM conflict analyzer because the Aptos Move memory model is quite different from EVM posing several challenges for a general specification derivation scheme.}
\revision{\subsection{Conflict Specifications for MoveVM} \label{sec:aptos-model}   
The Aptos memory model is quite different from EVM posing several challenges for a general specification derivation scheme. Aptos consists of a set of accounts $A$, where each account $a \in A$ has a set of modules $a.\texttt{mod}$ and a set of assets $a.\texttt{asset}$. A module is just a code in the Move language. Each asset $B$ is a struct defined in a module and can only be manipulated by the defining module. For example, as Aptos coin is defined by the aptos\_coin module, any module in need of transferring their Aptos coins needs to utilize public functions of aptos\_coin (if aptos\_coin did not provide such functions, manipulation would be possible only from inside the module).

To see why an EVM approach is not likely to work here, note that Move bytecode is stack-based, making storage tracking extremely difficult. Though the EVM opcode is also stack-based, we can bypass the problem in EVM with the help of a sophisticated decompiler such as Gigahorse, which derives a three-address code representation, whereas no similar tool is defined for Aptos. Even if such a tool existed for Aptos, there are other limitations. In EVM, a contract can only touch the balance or storage of any other account through specific opcodes. Thus, one can statically detect functions that never exit their contract. Moreover, P2P transactions can be clearly detected and can only touch their source and target account. In contrast, there is no clear bound on where a module can touch in the Move as it can move the assets of any account.

That is why as a proof of concept for the Aptos, we just selected two functions $TeviStar::transfer,aptos\_coin::transfer$ for deriving the specifications. Each of these functions is only supposed to transfer an asset from a source to a target. Let $T_1$ send some assets from $s_1$ to $t_1$ and similarly $T_2$ send some assets from $s_2$ to $t_2$. We only need to check $s_1,s_2,t_1,t_2$ are all different accounts to ensure that $T_1,T_2$ do not yield any memory conflicts.

\begin{algorithm}[!htbp]
\footnotesize
\caption{Deriving the Aptos's MoveVM specifications.}

\SetKwFunction{getCset}{get\_\Cset{}}
\SetKwProg{Fn}{function}{:}{}
\Fn{\getCset($T$)}{
    \uIf{$T.dest \notin \{Aptos,TeviStar\}$ and $T.function \neq transfer$}{
    \Return $\emptyset$.
    }
    
    $\Cset{} \gets \emptyset$.
    
    \tcp{Loop over transactions before $T$ in the \preset{} order.}
    
    \ForAll{$\{T' | T'.index<T.index\}$}{
            \uIf{$|\{T.signer,T'.signer,T.args[0],T'.args[0]\}| = 4$}{
                $\Cset{} \gets \emptyset$.
        }
    }
}
\end{algorithm}

Analyzing the miss-predictions of conflict specifications on real-world workloads is an important direction for future work. Such an analysis could help characterize false positives and false negatives in conflict specifications by the analyzer, reveal underlying workload patterns (e.g., contention hotspots, skewed access distributions), and ultimately guide the refinement of conflict detection heuristics or the development of more adaptive execution models.

\subsection{Conflict Analyzer Integration with MoveVM for Parallel Execution}
\label{subsec:conf_analyzer_movevm_exp}
We now present the microbenchmarks for the \emph{full} MoveVM implementation, leveraging the conflict specification provided by the conflict analyzer on both synthetic workloads. As described in \cref{Sec:experiments}, we implemented our \BTM{} algorithm on top of the Diem \BlockSTM{} and employed the same workload as previously described in \Cref{subsec:movevm_analysis}. The number of lines of code (Rust) for the Diem parallel-execution repository is 2391 and our MoveVM implementation is 4311 excluding the conflict analyzer. Overall, we added 2063 lines of code.

\noindent\textbf{Implementation Details.}
As discussed earlier in \Cref{subsec:conf_analyzer_evm_exp}, we introduced several implementation-level modifications to the \saBlockSTM{} executor, resulting in an optimized version that we call the \osaBlockSTM{}. Here, we compare the performance of \iBTM{} with \BSTM{} and the baseline sequential execution.

\begin{table}[!tb]
    \centering
    \caption{Conflict analyzer integrated with \osaBlockSTM{} for parallel execution in MoveVM; metrics from conflict analysis phase, block composition, and execution performance across different approaches.}
    \label{tab:conflict-analysis_wc_et_move}
    \resizebox{\textwidth}{!}{%
    \begin{tabular}{|r|r|c|c|c|}
        \hline
        \multicolumn{2}{|c|}{\multirow{2}{*}{\centering \textbf{Metric}}}
        & \multicolumn{3}{c|}{\textbf{Synthetic 10k Blocks}} \\
        \cline{3-5}
         \multicolumn{2}{|c|}{}
        & {\textbf{P2P Transfers}} 
        & {\textbf{Batch Transfers}} 
        & {\shortstack{\textbf{Generic} \textbf{(P2P + Batch)}}}
        \\
        \hline\hline

        \multirow{2}{*}{\shortstack{\textbf{Conflict}\\\textbf{Analysis Phase}}}
            & \textbf{Conflict Analysis Time} 
            & 1.04 ms 
            & 85.44 $\mu$s 
            & 638.847 $\mu$s 
            \\
            & \textbf{Specification Size} 
            & 3348 
            & 0 
            & 2243 
            \\
        \hline

        \multirow{5}{*}{\textbf{Block State}} 
            & \textbf{Block Size} 
            & 10000 
            & 10000 
            & 10000 
            \\
            & \textbf{P2P Transfers} 
            & 10000 
            & 0 
            & 0 
            \\
            & \textbf{Batch Transfers} 
            & 0 
            & 10000 
            & 0 
            \\
            & \textbf{Generic (P2P + Batch)}
            & 0 
            & 0 
            & 10000 
            \\
            & \textbf{Dependency Count} 
            & 6652 
            & 0 
            & 2757 
            \\
        \hline

        \multirow{3}{*}{\shortstack{\textbf{Average}\\\textbf{Execution Time}\\\textbf{(16 Threads)}}} 
            & \textbf{Sequential} 
            & 3302 ms 
            & 4551 ms 
            & 3895 ms 
            \\
            & \textbf{\textcolor{teal!80}{\BSTM{}}} 
            & \textcolor{teal!80}{156 ms} 
            & \textcolor{teal!80}{249 ms} 
            & \textcolor{teal!80}{203 ms}
            \\
            & \textbf{\textcolor{orange!80}{iBTM}}
            & \textbf{\textcolor{orange!80}{154 ms}}
            & \textbf{\textcolor{orange!80}{236 ms}}
            & \textbf{\textcolor{orange!80}{200 ms}}
            \\
        \hline

    \end{tabular}%
    }
\end{table}

\noindent
\textbf{Testing on Synthetic Workloads.}
\input{figs/plot-MoveVM-synt-integrated-main}
{\Cref{fig:movevm_batch_generic} presents the analysis of \iBTM{} implementation for MoveVM on synthetic workloads: P2P transfers, batch transfers, and a generic contract that combine P2P and batch transfer transactions in equal proportion. \Cref{fig:movevm_batch_generic} shows throughput and latency results for MoveVM across different synthetic workloads leveraging specifications generated by the integrated conflict analyzer. We compared the performance of \iBTM{} with the baseline sequential execution and the state-of-the-art \BlockSTM{} parallel execution while varying the thread count from 2-32, with a block size of 10k transactions. 

In the low-conflict P2P workload, \Cref{fig:movevm_batch_generic}a, where $\alpha = 0.1$, \iBTM{} scales from 12k-77k~tps, reducing execution latency 3.20$\times$. \BlockSTM{} achieves a comparable 73k~tps at 32 threads. In high-conflict P2P, where $\alpha = 1.0$ (\Cref{fig:movevm_batch_generic}a), throughput is modestly from 11k-14k~tps. However, in \Cref{fig:movevm_batch_generic}b, the batch transfer workload when $\alpha = 0.1$ exhibits better scalability, with \iBTM{} increasing from 7k-50k~tps, outperforming \BlockSTM{}’s 47k~tps. In the generic mixed workload with $\alpha = 0.1$ (\Cref{fig:movevm_batch_generic}b) \iBTM{} again scales well, improving from 9k-63k~tps. Overall, \iBTM{} consistently outperforms \BlockSTM{} and sequential execution, demonstrating the benefits of leveraging conflict-specification for optimistic parallel execution. 
}
}

    \section{Discussion and Concluding Remarks}\label{sec:conc}

The outstanding objective of this paper is to expose the significance of a smart contract parallel execution methodology that prioritizes leveraging transactional conflict specifications as input for the parallel executor. We demonstrated how this methodology can be constructive by implementing a state-of-the-art EVM parallel execution engine. We also demonstrated the methodology for the MoveVM and the empirical advantages it entails. We remark that we chose to exhibit this methodology for the \rwOblivious{} execution models of EVM and MoveVM when it would be more straightforward to implement our algorithms for the \rwAware{} execution model like Solana, since it does not require an explicit \emph{conflict analyzer} tool. We leave that for a future iteration of the paper.

\noindent\textbf{Localized Fee Markets.} Our parallel execution methodology possibly lays the foundation for more localized dynamic \emph{gas fee} marketplaces; if dependencies are specified a priori, transactions occurring in a congested contract of the blockchain state might be processed separately from others to prevent a localized state hotspot from increasing fees for the whole blockchain network. For example, a popular \emph{non fungible token (NFT)} mint could create a large number of transaction requests in a short period of time~\cite{nftMintingCongestion}. A \rwAware{} blockchain can detect state hotspots upfront, such as the NFT minting event, to rate limit and charge a higher fee for transactions containing them~\cite{HeliusSolanaLFM2024,solanaFeeMarketplace,SolanaLFM2024}. This enables ordinary transactions to execute promptly, while transactions related to the minting process are prioritized independently based on the total gas associated with them and resulting congestion.

    \bibliographystyle{plainurl}
    \bibliography{references2}
    \clearpage
\appendix
\noindent
\textbf{\Large Appendix}

\vspace{.2cm}
\noindent
This section is organized as follows:
\begin{table}[ht]
\vspace{-.24cm}
\label{tab:ap-Org}
\resizebox{.65\textwidth}{!}{%
    \begin{tabular}{l  l }
        

        
        \Cref{sec:conf_spec_evm}:
        & 
        \begin{tabular}
            [c]{@{}l@{}}
            Conflict Specifications for EVM
        \end{tabular} 
        \\
        \Cref{subsec:evm_conf_analyzer}:
        & 
        \begin{tabular}
            [c]{@{}l@{}}
            The EVM Conflict Analyzer
        \end{tabular} 
        \\
        \Cref{subsec:evm_conf_analyzer_exp}:
        & 
        \begin{tabular}
            [c]{@{}l@{}}
            Microbenchmarks for EVM Conflict Analyzer
        \end{tabular} 
        \\

    \end{tabular}%
    }
\end{table}



\section{Conflict Specifications for EVM}\label{sec:conf_spec_evm}
We previously discussed a limited implementation of the conflict analyzer for EVM \BTM{}, in the main draft \Cref{sec:spec}, which captured only a subset of transaction types. In this section, we present another implementation that provides a more comprehensive and robust form of conflict analysis. This stronger version generalizes to all transaction types and enables more accurate detection of execution dependencies across diverse workloads.
We refer to this as the \emph{strong conflict analyzer} which derives access specifications for public entry functions statically at the time of deployment of smart contracts using data-flow analysis on the smart contract code. This is a one-time computation (done possibly during blockchain \emph{epoch} changes or periodically), enabling efficient computation of conflict specifications for a block during transaction execution.

\subsection{{The Strong Conflict Analyzer for EVM}}\label{subsec:evm_conf_analyzer}
\begin{algorithm}[!htb]
\scriptsize
\caption{Derives the conflict specifications given as input to \Cref{algo:sdag,algo:saBlockSTM}; \textcolor{blue}{Strong Mode-lines~\ref{alg:alg1l1_apn}-\ref{alg:alg1l31_apn}} and Weak Mode-lines~\ref{alg:alg2l34_apn}-\ref{alg:alg2l66_apn}).}
\label{spec:algo_apn}
\KwIn {A transaction $T_j$ in a \preset{} order $T_1\to,\ldots,\to T_n$.}
\KwOut {$S \subseteq \{T_i | i<j,\Rset(T_j) \cap \Wset(T_i) = \emptyset\}$.}
\KwData{
$labels$: A map initialized to empty and filled during the preprocessing.
}
\SetKw{Continue}{continue}
\begin{minipage}[t]{0.465\textwidth}
\LinesNotNumbered
\SetNlSty{}{}{}
\begingroup\color{blue}
\vspace{-6pt}
\SetKwFunction{callGraph}{callGraph}
\SetKwProg{Fn}{Fun}{:}{}
\Fn{\callGraph($c^{TAC}$)}{\label{alg:alg1l1_apn}
\nlset{2} $V \gets$ functions in $c^{TAC}$
    \label{spec:algo:callgraph_apn}
    
    \nlset{3} $E \gets \emptyset$
    
    \nlset{4} 
    \ForAll{$v \in V$}{
    \nlset{5}
        \ForAll{$w \in V$ where $v$ code has a call to $w$}{
           \nlset{6} $E \gets E \cup \{w\}$
        }
    }
    
    \nlset{7} \Return $G = (V, E)$
}

\SetKwFunction{canReach}{canReach}
\SetKwProg{Fn}{Fun}{:}{}
\Fn{\canReach($G,v,opcodes$)}{
\nlset{8}
    \ForAll{$w \in V$ where $w$ can reach $v$ in $G$}{
        \nlset{9}\If{$w$ contains an opcode in $opcodes$}{
        \nlset{10} \Return true
        }
    }
    \nlset{11} \Return false.
}

\SetKwFunction{prep}{preprocess}
\SetKwProg{Fn}{Fun}{:}{}
\nlset{12}
\label{spec:algo:prep_apn}
\Fn{\prep}{
   \nlset{13} $exitOpcodes \gets$ \{CALL, DELEGATECALL, SELFDESTRUCT, CREATE,CREATE2\}
    
    \nlset{14}$staticExitOpcodes\gets$\{STATICCALL, \\BALANCE\}
    
    \nlset{15} \ForAll{contracts $c$}{
    \nlset{16} $c^{TAC} \gets$ $c$ proper representation
    
    \nlset{17} $G = (V, E) \gets \callGraph(c^{TAC})$

        \nlset{18} \ForAll{$v \in V$}{
            \nlset{19} $sig \gets v.\texttt{sig}$

            \nlset{20} \label{spec:algo:tokenlabel_apn}
            \uIf{$sig \in \{transfer, transferFrom, approve\}$ in \ERC{}}{ 
                \nlset{21} $labels[c,sig] \gets sig$.
            }
            \nlset{22}\uElseIf{\canReach($G, v, exitOpcodes$)}{
                \nlset{23} $labels[c, sig] \gets exitsContract$
            }
            \nlset{23}\uElseIf{\canReach($G, v, staticExitOpcodes$)}{
                \nlset{24} $labels[c, sig] \gets staticExitsContract$
            }
            \nlset{25}\Else{
                \nlset{26} $labels[c, sig] \gets insideContract$
            }
        }
    }

\SetKwFunction{getLabel}{getLabel}
\SetKwProg{Fn}{Fun}{:}{}
\nlset{27}
\Fn{\getLabel($T$)}{
    \label{spec:algo:label_apn}
        \nlset{28}
        \If{$T_i.dest \in A_U$}{
            \nlset{29} \Return $SimplePayment$
        }
        \nlset{30} \Else{
        \nlset{31} \Return $labels[T_j.dest, T.sig]$ \label{alg:alg1l31_apn}
        }
    }
}
\endgroup
\SetKwFunction{isToken}{is\ERC{}}
\SetKwProg{Fn}{Fun}{:}{}
\nlset{32}\Fn{\isToken($label$)}{\label{alg:alg2l32_apn}
    \nlset{33} \Return $label\in\{transfer,transferFrom,approve\}$
}
\end{minipage}
\hfill
\begin{minipage}[t]{0.419\textwidth}
\LinesNotNumbered
\SetNlSty{}{}{}
\SetKwFunction{ercAccounts}{erc20Accounts}
\SetKwProg{Fn}{Fun}{:}{}
\nlset{34}
\Fn{\ercAccounts($T$)}{\label{alg:alg2l34_apn}
\tcp{Get the affected accounts in an \ERC{} transaction; for example, a \emph{transferFrom()} includes the addresses \emph{from}, \emph{to}, and the transaction sender. As long as two \ERC{} transfers have disjoint sets of affected accounts, they are independent.}
    \nlset{35} \uIf{$label \in \{transfer,approve\}$}{
    \nlset{36} \Return $\{T.\texttt{origin}, T.\texttt{calldata}.\texttt{to}\}$
    }
    \nlset{37} \uElseIf{$label = transferFrom$}{
    \nlset{38} \Return $\{T.\texttt{origin}, T.\texttt{calldata}.\texttt{from}, T.\texttt{calldata}.\texttt{to}\}$
    }
}

\SetKwFunction{weakLabel}{getWeakLabel}
\SetKwProg{Fn}{Fun}{:}{}
\nlset{39}
\Fn{\weakLabel($T$)}{
    \nlset{40} \If{$T.dest \in A_U$}{
    \nlset{41} \Return $SimplePayment$.
    }
    \nlset{42} $sig \gets T.\texttt{sig}$

    \nlset{43} 
    \If{$sig \in \{transfer,transferFrom,approve\}$ in \ERC{}}{ 
        \nlset{44} \Return $T.\texttt{sig}$
    }   
    \nlset{45} \Else{
        \nlset{46} \Return $exitsContract$
    }
}

\SetKwFunction{rfconset}{findCSet}
\SetKwProg{Fn}{Fun}{:}{}
\nlset{47}
\Fn{\rfconset{$T_j$}}{
    \nlset{48} $label_j \gets \getLabel(T_j)$
    
    \nlset{49} \If{$label_j = exitsContract$ or $label_j = staticExitsContract$}{
        \nlset{50} \Return $\emptyset$
    }
    
    \nlset{51} $Output \gets \emptyset$
    
    \nlset{52} \ForAll{$1 \leq i < j$}{
        \nlset{56} $label_i \gets \getLabel(T_i)$
        }
        
        \nlset{57} \uIf{$T_i.\texttt{origin} = T_j.\texttt{origin}$ or $label_i = exitsContract$}{
        \nlset{58} \Continue
        }
        \nlset{59} 
        \label{spec:algo:token_apn}
        \uIf{$T_i.dest = T_j.dest$ and \isToken($label_i$ and \isToken($label_j$))}{
            \nlset{60} $affectedAccounts_i$ $\gets$ \ercAccounts($T_i$)
            
            \nlset{61} $affectedAccounts_j$ $\gets$ \ercAccounts($T_j$)

            \nlset{62} \uIf{$affectedAccounts_i \cap affectedAccounts_j = \emptyset$}{
                \nlset{63} $Output \gets Output \cup \{T_i\}$
            }
        }
        \nlset{64} 
        \label{spec:algo:distinct}
        \uElseIf{$|\{T_i.\texttt{origin}, T_j.\texttt{origin}, T_i.\texttt{dest}, T_j.\texttt{dest}\}| = 4$}{ 
            \nlset{65} $Output \gets Output \cup \{T_i\}$
        }
    }
    
    \nlset{66} \Return $Output$\label{alg:alg2l66_apn}

\end{minipage}
\end{algorithm}

\Cref{spec:algo_apn} underapproximates a conflict independence set given a sequence of transactions with a preset order. \Cref{algo:sdag} uses such a set to construct the transaction dependency graph. \Cref{algo:saBlockSTM} also uses this set to find the transactions independent from all the previous ones to execute concurrently. Of course, the set is an underapproximation of the actual conflict indepndence between the transactions as it tries to derive these relations \textit{statically}, that is before the actual execution and also independent of the current state of the blockchain.

\noindent\textbf{The Strong Mode Preprocessing.}
In the strong mode, before the actual conflict set derivation occurs, we assume that a preprocessing step has been done on the whole Ethereum ecosystem (in a production level implementation, such preprocessing should occur periodically for the newly deployed contracts). The first step in such a preprocess is some representation of the code enabling us to derive a function call graph of the contract. We will expand further on this in discussing the implementation. 

Then, the call graph is derived based on $c_{TAC}$ as in line~\ref{spec:algo:callgraph_apn}. During the preprocessing (line~\ref{spec:algo:prep_apn}), chosen special \ERC{} functions (transfer, transferFrom, approve) will be labeled according to the functions. For all other ones, we have assigned special opcodes $exitOpcodes$ and $staticExitOpcodes$ denoting respectively whether a contract can write to or read from other contracts' data. We label each function of a contract depending on whether it can reach those opcodes or not ($exitsContract$ for when the function can write to other accounts, $staticExitsContract$ for when it can only read from other accounts and $insideContract$ otherwise).

\noindent\textbf{Strong Mode Example.}
Now, we go over a run of the strong mode algorithm for transactions interacting with contracts Token and Wallet (cf. \cref{lst:cap0}) here which are very simple versions of actual token or wallet contracts. First notice that the function addToWallet contains none of the special opcodes ($CALL,STATICCALL$,...) and therefore should be labeled as $insideContract$. The transfer, however, has the signature of an \ERC{} transfer and thus is labeled as $transfer$. The withdraw function, however, contains an Ether transfer which at low level is translated to a $CALL$ and should be labeled as $exitsContract$. Similarly function $Token.turnEtherToToken$ is labeled as $staticExitsContract$.
\begin{itemize}
    \item $Wallet.withdraw$ is labeled as $exitsContract$ and thus is overapproximated to have read-from conflict with \textit{any} other transaction (but this is not guaranteed to be an actual read-from conflict as we are calculating an overapproximation).
    
    \item If $u_1$ initiates a transaction to call $Token.transfer$ and $u_2 \neq u_1$ initiates a transaction to call $Wallet.addToWallet$, these two have no read-from conflict. This function addToWallet as its label guarantees they will access nothing outside the accounts of the transaction origin and the contract. For the transfer, we just assume that the functions with \ERC{}-like functions faithfully implement the contract and thus will not deal with any other contract (except perhaps the proxy contract for the main \ERC{} contract).
    
    \item Similar to the previous part, $u_1$ initiates a transaction $T_1$ to call $Token.turnEtherToToken$ and $u_2$ initiates a transaction $T_2$ to call $Wallet.addToWallet$. Here, if $T_1$ is placed before $T_2$ in the preset order, they are overapproximated to have read-from conflict. But if $T_2$ comes first in the order, they are guaranteed not have read-from conflict. This is because $T_2$ only reads/modifies $u_2,Wallet$ and $T_1$ only modifies $u_1,Token$.
\end{itemize}
On the other hand, for any two different accounts $u_1,u_2$ calling transfer and addToWallet functions, there is no read-from conflict. This is evident as both of these functions do not exit their respective contract.

Note that we also make some assumptions about the decompiler in the function in line~\ref{spec:algo:label_apn}. This is used in the process exactly to ask one question: If there is no possible path from the function entry point to a certain opcode, then that opcode can never be reached through the execution of the bytecode with the same entry point. Thus, the following assumption about the decompiler is necessary.

\begin{proposition} \label{gigahorse:assumption}
    Let a contract $C$ and four bytes signature $s$ of a function $f$ be given. Refer to the proper representation of $c$ for callgraph as $c_{TAC}$. If there is no possible path from $f$ to an opcode $O$ in $c_{TAC}$, no transaction calling $c$ with signature $s$ can reach an opcode $O$.   
\end{proposition}

\begin{proposition} \label{spec:algo:access}
    Let $T$ be a transaction where $T.dest \in A_C$. Then if $T.calldata$ specifies an existing function signature in the contract,
    \begin{itemize}
        \item \label{case:exit}
        $T$ can only write to accounts other than $T.origin, T.dest$ if it is not a simple payment and some opcode in $exitOpcodes$ is reachable from the entry point.
        \item \label{case:staticexit}
        $T$ can only read from accounts other than $T.origin, T.dest$ if it is not a simple payment and some opcode in $exitOpcodes \cup staticExitOpcodes$ is reachable from the entrypoint.
    \end{itemize}
    where \(exitOpCodes = \{CALL,SELFDESTRUCT,CREATE,CREATE2\},\newline  staticExitOpcodes = \{STATICCALL, BALANACE\}\).
\end{proposition}

\begin{theorem}
    If the \Cref{spec:algo_apn} outputs an independence set $\mathcal{T}$ for a transaction $T_j$, then for each $T \in \mathcal{T}$ in the \preset{} order before $T_j$, $\Rset(T_j) \cap \Wset(T) = \emptyset$.
\end{theorem}
\begin{proof}
Throughout the proof, we assume as \cref{gigahorse:assumption} about $c^{TAC}$. Let $T \in \mathcal{T}$ in the \preset{} order be before $T_j$. First note that by the algorithm, if $T$ label is in $\{staticExitsContract,exitsContract\}$, then it is overapproximated to have read-from conflict with every transaction before in the preset order. So we can assume $T$ is either a simple payment or labeled $insideContract$ or it is an \ERC{} transaction. Then (according to the definition of simple payment and \cref{spec:algo:access} and line~\ref{spec:algo:label_apn}) $T$ only modifies/reads from $\{T.origin,T.dest\}$ which is guaranteed to be disjoint from $\{T_j.\texttt{origin}, T_j.\texttt{dest}\}$ by line~\ref{spec:algo:distinct}. Therefore by line~\ref{spec:algo:label_apn} and Proposition~\ref{spec:algo:access}, $T_j$ only modifies $\{T_j.\texttt{origin},T_j.\texttt{dest}\}$ which is disjoint from the read set of $T$ implying that $T$ does not have a read-from conflict with $T_j$. 

If both $T,T_j$ are \ERC{} transactions interacting with the same contract ($T.\texttt{dest} = T_j.\texttt{dest}$) calling by Proposition ~\ref{case:wellformed} and line~\ref{spec:algo:token_apn}, they are guaranteed to manipulate the tokens and accesses of a disjoint set of accounts.
\end{proof}

\subsection{Microbenchmarks for EVM Conflict Analyzer}\label{subsec:evm_conf_analyzer_exp}
\begin{table}[!htb]
    \caption{Weak Mode Specification Derivation Algorithm, the numbers reported are each average for all blocks in the corresponding range.}
    \label{tab:weak-specs}
    \centering
    \resizebox{1\columnwidth}{!}{%
        \begin{tabular}{| c | c | c | c | c | c |}
        \hline
        
        \multirow{2}{*}{\textbf{~~Block Range}~} & 
        \multirow{2}{*}{\textbf{~Block Size~}}  & 
        \multirow{2}{*}{\textbf{~Time (ms)~}}   & 
        \multirow{2}{*}{\parbox[c]{2.55cm}{\centering \textbf{Independent}\\\textbf{Tuples Fraction}}} &
        \multirow{2}{*}{\parbox[c]{2.85cm}{\centering \textbf{Simple Payment}\\\textbf{Fraction}}} &
        \multirow{2}{*}{\parbox[c]{1.55cm}{\centering \textbf{\ERC{}}\\\textbf{Fraction}}}\\
        & & & & &
        \\\hline\hline
        
        20320000-20320499 & 
        175.6 & 
        0.53 & 
        0.39 &
        0.31 &
        0.31\\
        
        15537293-15537392 & 
        149.0 & 
        0.17 & 
        0.06 &
        0.17 &
        0.07\\
        
        15537394-15537493 & 
        115.7 & 
        0.13 & 
        0.11 &
        0.22 &
        0.09\\
        
        4605067-4605166 & 
        71.1 & 
        0.20 & 
        0.41 &
        0.52 &
        0.16\\
        
        4605168-4605267 & 
        78.6 & 
        0.23 & 
        0.39 &
        0.53 &
        0.12\\
        \hline
        \end{tabular}
    }
\end{table}
\begin{table}[!htb]
    \caption{Strong Mode Specification Derivation Algorithm}
    \label{tab:strong-specs}
    \centering
    \resizebox{1\columnwidth}{!}{%
        \begin{tabular}{| c | c | c | c | c | c | c | c |}
        \hline
        
        \multirow{2}{*}{\textbf{Block Range}} & 
        \multirow{2}{*}{\textbf{Block Size}}   & 
        \multirow{2}{*}{\parbox[c]{2.25cm}{\centering \textbf{Decompilation}\\\textbf{Time (s)}}} & 
        \multirow{2}{*}{\parbox[c]{1.75cm}{\centering \textbf{Preprocess Time} (ms)}} &
        \multirow{2}{*}{\parbox[c]{1.53cm}{\centering \textbf{Analysis}\\\textbf{Time} (ms)}} &
        \multirow{2}{*}{\parbox[c]{1.2cm}{\centering \textbf{Tuples}\\\textbf{Fraction}}} &
        \multirow{2}{*}{\parbox[c]{2.85cm}{\centering \textbf{Simple Payment}\\\textbf{Fraction}}} &
        \multirow{2}{*}{\parbox[c]{1.6cm}{\centering \textbf{\ERC{}}\\\textbf{Fraction}}}\\
        & & & & & & &
        \\\hline\hline
        
        20320000-20320500 & 
        175.6 
        &
        12.7
        &
        162.8
        &
        0.65
        &
        0.41
        &
        0.31
        &
        0.31\\

        15537293-15537392
        &
        149.0
        &
        16.1
        &
        275.2
        &
        0.52
        &
        0.25
        &
        0.17
        &
        0.07\\

        15537394-15537493
        &
        115.7
        &
        13.1
        &
        211.3
        &
        0.44
        &
        0.37
        &
        0.22
        &
        0.09\\

        4605067-4605166
        &
        71.1
        &
        1.8
        &
        30.4
        &
        0.22
        &
        0.44
        &
        0.52
        &
        0.16\\

        4605167-4605266
        &
        78.6
        &
        2.1
        &
        30.4
        &
        0.28
        &
        0.43
        &
        0.53
        &
        0.12
        \\\hline
        \end{tabular}
    }
\end{table}

\noindent\textbf{Strong Mode Preprocessing.} We construct the call graph by using the three address code of the contract bytecode. To do so, decompilation is done for each contract by the state of the art decompiler for EVM bytecode Gigahorse \cite{grech2019gigahorse}. The preprocessing time is dominated by decompilation with orders of magnitude. This is natural because Gigahorse needs to run heavy logic programming analyses for its task of decompilation leading to relatively slow time. Assuming that this is done once for the whole blockchain and after that only periodically for the newly deployed contracts, this should not cause a major problem given that contract deployments whether by other contracts or EOA's represent only a small fraction of all transactions.

The benchmarks used for evaluation are five time ranges of Ethereum blocks as reported in \Cref{tab:weak-specs,tab:strong-specs}. The evaluations is performed on an AMD Ryzen 9 with 16 cores. Each of the measures provided is an average for blocks all over the corresponding range.

\noindent\textbf{Comparison to Weak Mode.}
As expected, the weak method speed outperforms the strong one in all rows. One can see that the more significant time gaps between the two methods occurs in rows $2,3$. This is again due to the lower ratio of simple payments and \ERC{} transactions leading to more aggressive pruning for the weak mode. A less important factor leading to slower speed for strong version can be the step of accessing the precomputed labels. For our experiments, this is not a bottleneck, however, if the precomputation is saved for the whole ecosystem the access time of the data structure can possibly cause a noticeable slowdown.

Similarly, the fraction of tuples derived by strong method is higher which is consistent with the definitions of both methods. Again, the highest gap is in rows $2,3$ where the ratio of simple payments and \ERC{} transactions is lower since these are the only transactions the weak method can reason about.

\noindent\textbf{Strong Mode Preprocessing Evaluation.}
For each block range, the algorithm starts with no preprocessed information but accumulates the preprocessing through that range. Thus, the cost of the decompilation/precomputation can be higher in first steps (up to 1-2 minutes) and gradually decreases as a lot of contracts have already been preprocessed.

In \Cref{tab:conflict-analysis-complete}, we compare the time taken by the strong and weak analyzers to compute specifications across different phases. This comparison helps us assess how these specifications can be effectively utilized in our integrated design. We evaluate both the number of specifications generated and various block-level statistics, including the count of native ETH transfers, \ERC{} transactions. Additionally, we examine whether the cost of generating these specifications can be amortized against the execution time of our methods. The dependency count for the Integrated analyzer includes only the dependencies added to the \osaBlockSTM{} scheduler after removing the transitive dependencies.

\begin{table}[ht]
    \centering
    \caption{Analysis of recent, historical, and large Ethereum blocks showing conflict statistics.}
    \label{tab:conflict-analysis-complete}
    \footnotesize
    \resizebox{\textwidth}{!}{%
    \begin{tabular}{|c|c c|c c|c c|c c c c|}
        \hline
        \multirow{2}{*}{\multirow{2}{*}{\parbox[c]{1.65cm}{\textbf{Block Number}}}} &
        \multicolumn{2}{c|}{\textbf{Analysis Phase}} &
        \multicolumn{2}{c|}{\textbf{Strong Case}} &
        \multicolumn{2}{c|}{\textbf{Weak Case}} &
        \multicolumn{4}{c|}{\textbf{Block Stats}} \\
        \cline{2-11}
        & \multirow{3}{*}{\parbox[c]{1cm}{\centering\textbf{Decomp}\\{(s)}}} 
        & \multirow{3}{*}{\parbox[c]{1.45cm}{\centering\textbf{Pre-proc}\\{(ms)}}} 
        & \multirow{3}{*}{\parbox[c]{1cm}{\centering\textbf{Conflict}\\{\textbf{Gen}}}} 
        & \multirow{3}{*}{\parbox[c]{1cm}{\centering\textbf{Spec}\\{\textbf{Size}}}} 
        & \multirow{3}{*}{\parbox[c]{1cm}{\centering\textbf{Conflict}\\{\textbf{Gen}}}} 
        & \multirow{3}{*}{\parbox[c]{1cm}{\centering\textbf{Spec}\\{\textbf{Size}}}} 
        & \multirow{3}{*}{\parbox[c]{.7cm}{\centering\textbf{Block}\\{\textbf{Size}}}} 
        & \multirow{3}{*}{\parbox[c]{.7cm}{\centering\textbf{ETH}\\{\textbf{Txns}}}} 
        & \multirow{3}{*}{\parbox[c]{.7cm}{\centering\textbf{\ERC{}}}} 
        & \multirow{3}{*}{\parbox[c]{.7cm}{\centering\textbf{SC}\\{\textbf{Txns}}}} \\ 
        & & & & & & & & & &\\ 
        & & & & & & & & & &\\\hline\hline

        \multicolumn{11}{|c|}{\multirow{2}{*}\centering\textbf{Ethereum 2.0}}
        \\\hline
        & & & & & & & & & &\\
        \textbf{15537400}  & 120.63  & 89.98  & 113.02 µs & 4531 & 26.59 µs & 190   & 178  & 14  & 6  & 158 \\
        \textbf{15537401}  & 6.79    & 13.16  & 36.85 µs  & 1200 & 15.16 µs & 36    & 70   & 9   & 0  & 61  \\
        & & & & & & & & & &\\
        \hline

        \multicolumn{11}{|c|}{\multirow{2}{*}\centering\textbf{CryptoKitties}}\\\hline
        & & & & & & & & & &\\
        \textbf{4605100}  & 3.81    & 11.21  & 17.86 µs  & 135  & 19.47 µs & 135   & 53   & 11  & 6  & 36 \\
        \textbf{4605101}   & 5.63    & 7.61   & 17.62 µs  & 230  & 17.52 µs & 207   & 34   & 15  & 8  & 11 \\
        & & & & & & & & & &\\
        \hline

        \multicolumn{11}{|c|}{\multirow{2}{*}\centering\textbf{Large Blocks}}\\\hline
        & & & & & & & & & &\\
        \textbf{17873752}  & 73.35   & 40.81  & 2.53 ms   & 54155 & 106.79 µs  & 50706 & 1189 & 1129 & 21 & 39 \\
        
        
        
        \textbf{17873654}  & 11.31   & 49.81  & 66.25 µs  & 1930  & 30.09 µs & 1691  & 136  & 41   & 18 & 77 \\
        & & & & & & & & & &\\
        \hline
    \end{tabular}%
    }
\end{table}

\noindent\textbf{Key Takeaways.}
We have developed approaches to derive the conflict independence set with or without preprocesing.
The significant difference between the two approaches can be seen in blocks with low ratio of simple payments and \ERC{} transactions, where the weak method is much faster but also derives a much lower ratio of independent tuples. The preprocessing time is dominated by the decompilation step. Though this does not cause any problem since it is done once for the whole ecosystem and after that only periodically for the small fraction of contract creating transactions.

\end{document}